\documentclass[runningheads]{llncs}
\usepackage{graphicx}
\usepackage{xcolor}
\usepackage{amssymb,amsmath}
\usepackage{wrapfig}
\usepackage{paralist}

\renewcommand{\Pr}{\mathbb{P}}
\newcommand{\E}{{\mathbb{E}}}
\newenvironment{sop}{\noindent{\em Sketch of proof.}}{\qed}

\begin{document}
	\title{Preferential attachment hypergraph with high modularity\thanks{\scriptsize The Appendix contains proofs, comments on the model implementation and results of further experiments with a real data.}}
	%
	%
	\author{Fr{\'e}d{\'e}ric Giroire\inst{1} \and Nicolas Nisse\inst{1}
	\and Thibaud Trolliet\inst{1} \and Ma{\l}gorzata Sulkowska\inst{1,2}}
	\authorrunning{F. Giroire et al.}
	%
	\institute{Universit{\'e} C{\^o}te d’Azur, CNRS, Inria, I3S, France
\\
		\and Wroc{\l}aw University of Science and Technology,
		Faculty of Fundamental Problems of Technology,
		Department of Fundamentals of Computer Science, Poland
}
	\maketitle              
	\begin{abstract}
		Numerous works have been proposed to generate random graphs preserving the same properties as real-life large scale networks. However, many real networks are better represented by hypergraphs. Few models for generating random hypergraphs exist and no general model allows to both preserve a power-law degree distribution and a high modularity indicating the presence of communities. We present a dynamic preferential attachment hypergraph model which features partition into communities. We prove that its degree distribution follows a power-law and we give theoretical lower bounds for its modularity. We compare its characteristics with a real-life co-authorship network and show that our model achieves good performances.
		We believe that our hypergraph model will be an interesting tool that may be used in many research domains in order to reflect better real-life phenomena.
	\end{abstract}
	\section{Introduction}
	
The area of complex networks concerns designing and analysing structures that model well large real-life systems. It was empirically recognised that the common ground of such structures are small diameter, high clustering coefficient, heavy tailed degree distribution and visible community structure \cite{Bol_Rio_chapter}. Surprisingly, all those characteristics appear, no matter whether we investigate biological, social, or technological systems. 
A dynamical growth in research in this field one observes roughly since 1999 when Barab{\'a}si and Albert introduced probably the most studied nowadays preferential attachment graph~\cite{BA_basic}. Their model is based on two mechanisms: growth (the graph is growing over time, gaining a new vertex and a bunch of edges at each time step) and preferential attachment (arriving vertex is more likely to attach to other vertices with high degree rather than with low degree).
It captures two out of four universal properties of real networks, which are a heavy tailed degree distribution and a small world phenomenon. 

A number of theoretical models were presented throughout last 25 years. Just to mention the mostly investigated ones: Watts and Strogatz (exhibiting small-world and high clustering properties~\cite{Watts_Strogatz}), Molloy and Reed (with a given degree sequence~\cite{Molloy_Reed}), Chung-Lu (with a given expected degree sequence~\cite{Chung_Lu_model}), Cooper-Frieze (model of web graphs~\cite{Cooper_Frieze}), Buckley-Osthus~\cite{Buck_Ost} or random intersection graph (with high clustering properties and following a power-law, \cite{rand_inters_models}). None of here mentioned graphs captures all the four properties listed in the previous paragraph, e.g., \cite{Watts_Strogatz} does not have a heavy tailed degree distribution, \cite{BA_basic} and \cite{Chung_Lu_model} models suffer from vanishing clustering coefficient \cite{Bol_Rio_chapter}, almost all of them do not exhibit visible community structure, i.e., have low \emph{modularity}. 

Modularity is a parameter measuring how clearly a network may be divided into \emph{communities}. 
It was introduced by Newman and Girvan in \cite{Newman_mod_def}. A graph has high modularity if it is possible to 
partition the set of its vertices into communities inside which the density of edges is remarkably higher than the density of edges between different communities.  Modularity is known to have some drawbacks (for thorough discussion check \cite{limits_resolution}). 
Nevertheless, today it remains a popular measure and is widely used in most common algorithms for community detection \cite{community_detection,Louv,Leiden}. It is well known that the real-life social or biological networks are highly modular \cite{Fortunato,Girvan_Newman}. 
At the same time simulations show that most of existing preferential attachment models have low modularity. Good modularity properties one finds in geometric models, like spatial preferential attachment graphs \cite{spatial,jacob_spatial}, however they use additionally a spatial metric. 

Finally, almost all the up-to-date complex networks models are graph models 
thus are able to mirror only binary relations. In practical applications $k$-ary relations (co-authorship, groups of interests or protein reactions) are often modelled in graphs by cliques which may lead to a profound information loss.

\noindent\textbf{Results.} Within this article we propose a dynamic model with high modularity by preserving a heavy tailed degree distribution and not using a spatial metric. Moreover, our model is a random hypergraph (not a graph) thus can reflect $k$-ary relations. Preferential attachment hypergraph model was first introduced by Wang et al. in \cite{Wang_hyper}. However, 
it was restricted just to a specific subfamily of uniform acyclic hypergraphs (the analogue of trees within graphs). The first rigorously studied non-uniform hypergraph preferential attachment model was proposed only in 2019 by Avin et al. \cite{ALP_hyper}. Its degree distribution follows a power-law. However, our empirical results indicate that this model has a weakness of low modularity (see Section \ref{sec:mod_exp}). To the best of our knowledge the model proposed within this article is the first dynamic non-uniform hypergraph model with degree sequence following a power-law and exhibiting clear community structure. We experimentally show that features of our model correspond to the ones of a real co-authorship network built upon Scopus database.

\noindent \textbf{Paper organisation.} Basic definitions are introduced in Sec.~\ref{sec:notation}. In Sec. \ref{sec:pa_hyper}, we present a universal preferential attachment hypergraph model which unifies many existing models (from classical Barab{\'a}si-Albert graph \cite{BA_basic} to Avin et al. preferential attachment hypergraph \cite{ALP_hyper}). In Sec. \ref{sec:sbm}, we use it as a component in a stochastic block model to build a general hypergraph with good modularity properties. Theoretical bounds for its modularity and experimental results on a real data are presented in Sec.~\ref{sec:modularity}. 
Further works are presented 
in Sec.~\ref{sec:summary}. 


	
	\section{Basic definitions and notation} \label{sec:notation}
	
We define a \emph{hypergraph} $H$ as a pair $H=(V,E)$, where $V$ is a set of vertices and $E$ is a set of hyperedges, i.e., non-empty, unordered multisets of $V$. We allow for a multiple appearance of a vertex in a hyperedge (self-loops). The degree of a vertex $v$ in a hyperedge $e$, denoted by $d(v,e)$, is the number of times $v$ appears in $e$. The cardinality of a hyperedge $e$ is  $|e| = \sum_{v \in e} d(v,e)$. The degree of a vertex $v \in V$ in $H$ is understood as the number of times it appears in all hyperedges, i.e., $\deg(v) = \sum_{e \in E} d(v,e)$. If $|e|=k$ for all $e \in E$, $H$ is said {\it $k$-uniform}.

We consider hypergraphs that grow by adding vertices and/or hyperedges at discrete time steps $t=0,1,2,\ldots$. The hypergraph obtained at time $t$ will be denoted by $H_t=(V_t,E_t)$ and the degree of $u \in V_t$ in $H_t$ by $\deg_t(u)$. By $D_t$ we denote the sum of degrees at time $t$, i.e., $D_t = \sum_{u \in V_t} \deg_t(u)$. As the hypergraph gets large, the probability of creating a self-loop can be well bounded and is quite small provided that the sizes of hyperedges are reasonably bounded.

$N_{k,t}$ stands for the number of vertices in $H_t$ of degree $k$. We say that the degree distribution of a hypergraph follows a \emph{power-law} if the fraction of vertices of degree $k$ is proportional to $k^{-\beta}$ for some exponent $\beta \geq 1$. Formally, we will interpret it as $\lim_{t \rightarrow \infty} \E\left[\frac{N_{k,t}}{|V_t|}\right] \sim c \cdot k^{-\beta}$ for some positive constant $c$ and $\beta \geq 1$. For $f$ and $g$ being real functions we write $f(k) \sim g(k)$ if $f(k)/g(k) \xrightarrow{k \rightarrow \infty} 1$.

\emph{Modularity} measures the presence of community structure in the graph.
Its definition for graphs introduced by Newman and Girvan in 2004 is given below.
\begin{definition}[\cite{Newman_mod_def}] \label{def:mod_graphs}
Let $G=(V,E)$ be a graph with at least one edge. For a partition $\mathcal{A}$ of vertices of $G$ define its modularity score on $G$ as
\[
q_{\mathcal{A}}(G) = \sum_{A \in \mathcal{A}} \left(\frac{|E(A)|}{|E|} - \left(\frac{vol(A)}{2|E|}\right)^2 \right),
\]
where $E(A)$ is the set of edges within $A$ and $vol(A) = \sum_{v \in A} deg(v)$. 
Modularity of $G$ is given by $q^*(G) = \max_{\mathcal{A}} q_{\mathcal{A}}(G)$.
\end{definition}
\noindent
Conventionally, a graph with no edges has modularity equal to $0$. The value $\sum_{A \in \mathcal{A}} \frac{|E(A)|}{|E|}$ is called an \textit{edge contribution} while $\sum_{A \in \mathcal{A}} \left(\frac{vol(A)}{2|E|}\right)^2$ is a \textit{degree tax}. A single summand of the modularity score is the difference between the fraction of edges within $A$ and the expected fraction of edges within $A$ if we considered a random multigraph on $V$ with the degree sequence given by $G$. One can observe that the value of $q^*(G)$ always falls into the interval $[0,1)$. 

Several approaches to define a modularity for hypergraphs can be found in contemporary literature. Some of them flatten a hypergraph to a graph (e.g., by replacing each hyperedge by a clique) and apply a modularity for graphs (see e.g. \cite{hyper_flat_2}). Others base on information entropy modularity \cite{hyper_entropy}. We want to stick to the classical definition from \cite{Newman_mod_def} and preserve a rich hypergraph structure, therefore we work with the definition proposed by Kami{\'n}ski et al. in \cite{Pralat_hyper_modularity}. 


\begin{definition} [\cite{Pralat_hyper_modularity}] \label{def:mod_hypergraphs}
Let $H=(V,E)$ be a hypergraph with at least one hyperedge. For $\ell \geq 1$ let $E_{\ell} \subseteq E$ denote the set of hyperedges of cardinality $\ell$. For a partition $\mathcal{A}$ of vertices of $H$ define its modularity score on $H$ as
\[
q_{\mathcal{A}}(H) = \sum_{A \in \mathcal{A}} \left(\frac{|E(A)|}{|E|} - \sum_{\ell \geq 1} \frac{|E_{\ell}|}{|E|} \cdot \left(\frac{vol(A)}{vol(V)}\right)^{\ell} \right),
\]
where $E(A)$ is the set of hyperedges within $A$ (a hyperedge is within $A$ if all its vertices are contained in $A$), $vol(A) = \sum_{v \in A} deg(v)$ and $vol(V) = \sum_{v \in V} deg(v)$. Modularity of $H$ is given by $q^*(H) = \max_{\mathcal{A}} q_{\mathcal{A}}(H)$.
\end{definition}
\noindent
A single summand of the degree tax 
is the expected number of hyperedges within $A$ if we considered a random hypergraph on $V$ with the degree sequence given by $H$ and having the same number of hyperedges of corresponding cardinalities.

We write that an event $A$ occurs \emph{with high probability} (whp) if the probability $\Pr[A]$ depends on a certain number $t$ and tends to $1$ as $t$ tends to infinity.
	
	\section{General preferential attachment hypergraph model} \label{sec:pa_hyper}

In this section we generalise a hypergraph model proposed by Avin et al. in \cite{ALP_hyper}. Model from \cite{ALP_hyper} allows for two different actions at a single time step - attaching a new vertex by a hyperedge to the existing structure or creating a new hyperedge on already existing vertices. We allow for four different events at a single time step, admit the possibility of adding more than one hyperedge at once and draw the cardinality of newly created hyperedge from more than one distribution. 
The events allowed at a single time step in our model $H_t$ are: adding an isolated vertex, adding a vertex and attaching it to the existing structure by $m$ hyperedges, adding $m$ hyperedges, or doing nothing. The last event ``doing nothing'' is included since later we put $H_t$ in a broader context of stochastic block model, where it serves as a single community. ``Doing nothing'' indicates a time slot in which nothing associated directly with $H_t$ happens but some event takes place in the other part of the whole stochastic block model.


\subsection{Model $\mathbf{H(H_0,p,Y,X,m,\gamma)}$}
General hypergraph model $H$ is characterized by six parameters. These are:
\begin{enumerate}
\item $H_0$ - initial hypergraph, seen at $t=0$;
\item $\mathbf{p} = (p_v,p_{ve},p_e)$ - vector of probabilities indicating, what are the chances that a particular type of event occurs at a single time step; we assume $p_v + p_{ve} + p_e \in (0,1]$; additionally $p_e$ is split into the sum of $r$ probabilities $p_e = p_e^{(1)} + p_e^{(2)} + \ldots + p_e^{(r)}$ which allows for adding hyperedges whose cardinalities follow different distributions;
\item $Y = (Y_0, Y_1, \ldots, Y_t, \ldots)$ - independent random variables, cardinalities of hyperedges that are added together with a vertex at a single time step;
\item $X = ((X_1^{(1)}, \ldots, X_t^{(1)}, \ldots), (X_1^{(2)}, \ldots, X_t^{(2)}, \ldots), \ldots, (X_1^{(r)}, \ldots, X_t^{(r)}, \ldots))$ - $r$ sequences of independent random variables, cardinalities of hyperedges that are added at a single time step when no new vertex is added; 
\item $m$ - number of hyperedges added at once;
\item $\gamma \geqslant 0$ - parameter appearing in the formula for the probability of choosing a particular vertex to a newly created hyperedge.
\end{enumerate}
Here is how the structure of $H = H(H_0,p,Y,X,m,\gamma)$ is being built. We start with some non-empty hypergraph $H_0$ at $t=0$. We assume for simplicity that $H_0$ consists of a hyperedge of cardinality $1$ over a single vertex. Nevertheless, all the proofs may be generalised to any initial $H_0$ having constant number of vertices and constant number of hyperedges with constant cardinalities.
`Vertices chosen from $V_t$ in proportion to degrees' means that vertices are chosen independently (possibly with repetitions) and the probability that any $u$ from $V_t$ is chosen is
\[
\Pr[u \textnormal{ is chosen}] = \frac{\deg_t(u) + \gamma}{\sum_{v \in V_t}(\deg_t(v)+\gamma)} = \frac{\deg_t(u) + \gamma}{D_t + \gamma|V_t|}.
\]
For $t \geqslant 0$ we form $H_{t+1}$ from $H_t$ choosing only one of the following events according to $\mathbf{p}$.
\begin{itemize}
\item With probability $p_v$: Add one new isolated vertex.
\item With probability $p_{ve}$: Add one vertex $v$. Draw a value $y$ being a realization of $Y_t$. Then repeat $m$ times: select $y-1$ vertices from $V_{t}$ in proportion to degrees; add a new hyperedge consisting of $v$ and $y-1$ selected vertices.
\item With probability $p_{e}^{(1)}$: Draw a value $x$ being a realization of $X_t^{(1)}$. Then repeat $m$ times: select $x$ vertices from $V_{t}$ in proportion to degrees; add a new hyperedge consisting of $x$ selected vertices.
\item $\ldots$
\item With probability $p_{e}^{(r)}$: Draw a value $x$ being a realization of $X_t^{(r)}$. Then repeat $m$ times: select $x$ vertices from $V_{t}$ in proportion to degrees; add a new hyperedge consisting of $x$ selected vertices.
\item With probability $1-(p_v+p_{ve}+p_{e})$: Do nothing.
\end{itemize}
We allow for $r$ different distributions from which one can draw the cardinality of newly created hyperedges. Later, when $H_t$ serves as a single community in the context of the whole stochastic block model, this trick allows for spanning a new hyperedge across several communities drawing vertices from each of them according to different distributions. This reflects some possible real-life applications. Think of an article authored by people from two different research centers. Our experimental observation is that it is very unlikely that the number of authors will be distributed uniformly among two centers. More often, one author represents one center, while the others are affiliated with the second one.

\subsection{Degree distribution of $\mathbf{H(H_0,p,Y,X,m,\gamma)}$}

In this section we prove that the degree distribution of $H = H(H_0,p,Y,X,m,\gamma)$ follows a power-law with $\beta > 2$. 
We assume that supports of random variables indicating cardinalities of hyperedges are bounded and their expectations are constant. This assumption is in accord with potential applications - think of co-authors, groups of interest, protein reactions, ect.

\begin{theorem} \label{thm:hyper_plaw}
Consider a hypergraph $H = H(H_0,\mathbf{p},Y,X,m,\gamma)$ for any $t>0$. Let $i \in \{1,\ldots,r\}$. Let $\E[Y_t] = \mu_0$, and $\E[X_t^{(i)}] = \mu_i$. Moreover, let $1 \leqslant Y_t < t^{1/4}$ and $1 \leqslant X_t^{(i)} < t^{1/4}$. Then the degree distribution of $H$ follows a power-law with 
$$ \beta = 2+\frac{\gamma \bar{V} + m \cdot p_{ve}}{\bar{D} - m \cdot p_{ve}},$$
where $\bar{V} = p_v+p_{ve}$ and $\bar{D} = m(p_{ve}\mu_0 + p_e^{(1)}\mu_1 + \ldots + p_e^{(r)}\mu_r)$ which are the expected number of vertices added per a single time step and the expected number of vertices that increase their degree in a single time step, respectively.
\end{theorem}

\begin{sop}
We prove that $\lim_{t \rightarrow \infty} \frac{\E[N_{k,t}]}{|V_t|} \sim \tilde{c} k^{-\beta}$ (determining the exact constant $\tilde{c}$). For this purpose we first show that it is sufficient to prove that $\lim_{t \rightarrow \infty} \frac{\E[N_{k,t}]}{t} \sim c k^{-\beta}$ (Lemma~\ref{lemma:Nkt_limit} in the Appendix). Let ${\mathcal{F}}_t$ be the $\sigma$-algebra associated with the probability space at time $t$. Let $Q_{d,k,t}$ denote the probability that a specific vertex of degree $k$ was chosen $d$ times to be included in new hyperedges at time $t$. 
Moreover, let $Z_t$ be the random variable chosen at step $t$ among $Y_t, X_t^{(1)}, \ldots, X_t^{(r)}$ according to $(p_v,p_{ve},p_e^{(1)},\ldots,p_e^{(r)})$.
For $t \geqslant 1$ we get that $\E[N_{0,t} |{\mathcal{F}}_{t-1}] = p_v + N_{0,t-1} Q_{0,0,t}$ 	and when $k \geqslant 1$:
	\[
		\E[N_{k,t} |{\mathcal{F}}_{t-1}]  = \delta_{k,m} p_{ve} + \sum_{i=0}^{\min\{k,mZ_t\}} N_{k-i,t-1} Q_{i,k-i,t},
	\]
	where $\delta_{k,m}$ is the Kronecker delta. The proof then follows from the tedious analysis of this recursive equation.
\end{sop}




Below we present a bunch of examples showing that our theorem generalises the results for the degree distribution of well known models.

\begin{example}[Barab{\'a}si-Albert graph model, \cite{BA_basic}]
In a single time step we always add one new vertex and attach it with $m$ edges \textnormal{(}in proportion to degrees\textnormal{)} to existing structure. Thus $p_v = 0$, $p_{ve}=1$, $p_e=0$, $\bar{V}=1$, $Y_t=2$, $\bar{D}=2m$, $\gamma=0$ and we get $\beta = 2+\frac{m}{2m-m} = 3.$
\end{example}

\begin{example}[Chung-Lu graph model, \cite{ChLu_book}]
In a single time step: we either (with probability $p$) add one new vertex and attach it with an edge (in proportion to degrees) to existing structure; otherwise we just add an edge \textnormal{(}in proportion to degrees\textnormal{)} to existing structure. Thus $p_v = 0$, $p_{ve}=p$, $p_e=1-p$, $\bar{V}=p$, $Y_t=2$, $r=1$, $X_t^{(1)}=2$, $\bar{D}=2$, $m=1$, $\gamma=0$ and we get $\beta = 2+\frac{p}{2-p}$.
\end{example}

%
%
\begin{example}[Avin et al. hypergraph model, \cite{ALP_hyper}]
In a single time step we either (with probability $p$) add one new vertex and attach it with a hyperedge of cardinality $Y_t$ \textnormal{(}in proportion to degrees\textnormal{)} to existing structure; otherwise we just add a hyperedge of cardinality $Y_t$ 
to existing structure. The assumptions on $Y_t$ and the sum of degrees $D_t$ are:
\begin{inparaenum}
\item $\lim_{t \rightarrow \infty} \frac{\E[D_{t-1}]/t}{\E[Y_t]-p_{ve}} = D \in (0,\infty)$,
\item $\E[|\frac{1}{D_t}-\frac{1}{\E[D_t]}|] = o(1/t)$,
\item $\E\left[\frac{Y_t^2}{D_{t-1}^2}\right] = o(1/t)$.
\end{inparaenum}
The result from \cite{ALP_hyper} states that the degree distribution of the resulting hypergraph follows a power-law with $\beta = 1 + D$. Note that in our model $\lim_{t \rightarrow \infty} \frac{\E[D_{t-1}]/t}{\E[Y_t]-p_{ve}} = \frac{\bar{D}}{\bar{D}-p_{ve}}$. Setting $p_v = 0$, $p_{ve}=p$, $p_e=1-p$, $\bar{V}=p$, $m=1$, $\gamma=0$ we get $\beta = 2+\frac{p_{ve}}{\bar{D}-p_{ve}} = 1 + \frac{\bar{D}}{\bar{D}-p_{ve}} = 1 + D$.
\end{example}

\begin{remark}
Even though our result from this section may seem similar to what was obtained by Avin et al., it is easy to indicate cases that are covered by our model but not by the one from \cite{ALP_hyper} and vice versa. Indeed, the model from \cite{ALP_hyper} admits a wide range of distributions for $Y_t$. In particular, as authors underline, three mentioned assumptions hold for $Y_t$ which is polynomial in $t$. This is the case not covered by our model (we upper bound $Y_t$ by $t^{1/4}$) but we also can not think of real-life examples that would require bigger hyperedges. Whereas we can think of some natural examples that break requirements from \cite{ALP_hyper} but are admissible in our model. Put $Y_t = 2$ if $t$ is odd and $Y_t = 3$ if $t$ is even. 
Then $\lim\limits_{\substack{t \rightarrow \infty \\ t \textnormal{ - even}}} \frac{\E[D_{t-1}]/t}{\E[Y_t]-p_{ve}} = \frac{5/2}{3-p_{ve}}$ and $\lim\limits_{\substack{t \rightarrow \infty \\ t \textnormal{ - odd}}} \frac{\E[D_{t-1}]/t}{\E[Y_t]-p_{ve}} = \frac{5/2}{2-p_{ve}}$
thus the limit $\lim_{t \rightarrow \infty} \frac{\E[D_{t-1}]/t}{\E[Y_t]-p_{ve}}$ does not exist. Whereas in our model we are allowed to put $r=2$, $p_e^{(1)}=p_e^{(2)}=1/2$, $X_t^{(1)} = 2$, $X_t^{(2)} = 3$ which  probabilistically simulates stated example.
\end{remark}

		\section{Hypergraph model with high modularity} \label{sec:sbm}

In this section we present a new preferential attachment hypergraph model which features partition into communities. To the best of our knowledge no mathematical model so far consolidated preferential attachment, possibility of having hyperedges and clear community structure. We prove that its degree distribution follows  a power-law. We denote our hypergraph by $G_t=(V_t,E_t)$. At each time step either a new vertex (\emph{vertex-step}) or a new hyperedge (\emph{hyperedge-step}) is added to the existing structure. The set of vertices of $G_t$ is partitioned into $r$ communities $V_t = C_t^{(1)} \mathbin{\dot{\cup}} C_t^{(2)} \mathbin{\dot{\cup}} \ldots \mathbin{\dot{\cup}} C_t^{(r)}$. Whenever a new vertex is added to $G_t$ it is assigned to the one of $r$ communities and stays there forever. 

\subsection{Model $\mathbf{G(G_0,p,M,X,P,\gamma)}$}

Hypergraph model $G$ is characterized by six parameters:
\begin{enumerate}
\item $G_0$ - initial hypergraph seen at time $t=0$ with vertices partitioned into $r$ communities $V_0 = C_{0}^{(1)} \mathbin{\dot{\cup}} C_{0}^{(2)} \mathbin{\dot{\cup}} \ldots \mathbin{\dot{\cup}} C_{0}^{(r)}$;
\item $p \in (0,1)$ - the probability of taking a vertex-step;
\item vector $M=(m_1, m_2, \ldots, m_r)$ with all $m_i$ positive, constant and summing up to $1$; $m_i$ is the probability that a randomly chosen vertex belongs to $C_t^{(i)}$; 
\item $d$-dimensional matrix $P_{r \times \ldots \times r}$ of hyperedge probabilities ($P_{i_1,i_2,\ldots,i_d}$ is the probability that communities $i_1, \ldots, i_d$ share a hyperedge); $d$ is the upper bound for the number of communities shared by a single hyperedge; 
\item $X = ((X_0^{(1)}, X_1^{(1)}, \ldots), (X_0^{(2)}, X_1^{(2)}, \ldots), \ldots, (X_0^{(d)}, X_1^{(d)}, \ldots))$ - $d$ sequences of independent random variables indicating the number of vertices from a particular community involved in a newly created hyperedge; 
\item $\gamma \geqslant 0$ - parameter appearing in the formula for the probability of choosing a particular vertex to a newly created hyperedge. 
\end{enumerate}

We build a structure of $G(G_0,p,M,X,P,\gamma)$ starting with some initial hypergraph $G_0$. Here $G_0$ consists of $r$ disjoint hyperedges of cardinality $1$. All vertices are assigned to different communities. `Vertices are chosen from $C_{t}^{(i)}$ in proportion to degrees' means that vertices are chosen independently (possibly with repetitions) and the probability that any $u$ from $C_{t}^{(i)}$ is chosen equals
\[
\Pr[u \textnormal{ is chosen}] = \frac{\deg_t(u) + \gamma}{\sum_{v \in C_{t}^{(i)}}{(\deg_t(v)+\gamma)}},
\]
($\deg_t(v)$ is the degree of $v$ in $G_t$). For $t \geqslant 0$, $G_{t+1}$ is obtained from $G_t$ as follows:
\begin{itemize}
\item With probability $p$ add one new isolated vertex and assign it to one of $r$ communities according to a categorical distribution given by vector $M$. 
\item Otherwise, create a hyperedge:
	\begin{itemize}
	\item according to $P$ select $N$ communities ($N$ is a random variable depending on $P$) that will share a hyperedge being created, say $C_{t}^{(i_1)}$, $C_{t}^{(i_2)}, \ldots, C_{t}^{(i_N)}$;
	\item assign selected communities to $N$ random variables chosen from $\{X_t^{(1)}$, $\ldots, X_t^{(r)}\}$ uniformly independently at random, say to $X_{t}^{(j_1)}$, $\ldots, X_{t}^{(j_N)}$;
	\item for each $s \in \{1,\ldots,N\}$ select $X_{t}^{(j_s)}$ vertices from $C_{t}^{(i_s)}$ in proportion to degrees;
	\item create a hyperedge consisting of all selected vertices.
	\end{itemize}
\end{itemize}


\subsection{Degree distribution of $\mathbf{G(G_0,p,M,X,P,\gamma)}$}

A power-law degree distribution of $G$ comes from the fact that each community of $G$ behaves over time as the hypergraph model $H$ presented in previous section. Thus the degree distribution of each community follows a power-law. 
For a detailed proof and experimental results on the degree distribution of a real-life co-authorship network check the Appendix.

\begin{theorem} \label{thm:degreed_G}
Consider a hypergraph $G = G(G_0,p,M,X,P,\gamma)$ for all $t>0$. Let $\E[X_t^{(i)}] = \mu_i$ and $1 \leqslant X_t^{(i)} < t^{1/4}$ for $i \in \{0,1,\ldots,r\}$. Then the degree distribution of $G$ follows a power-law with $\beta = 2 + \gamma \cdot \min_{j\in\{1,\ldots,r\}} \left\{ \bar{V}_j/\bar{D}_j \right\}$, where $\bar{V}_j$ is the expected number of vertices added to $C_{t}^{(j)}$ at a single time step and $\bar{D}_j$ is the expected number of vertices from $C_{t}^{(j)}$ that increase their degree at a single time step. I.e.,
\[
\beta = 2 + \frac{\gamma p }{(1-p)\frac{\mu_1 + \ldots + \mu_r}{r}} \cdot \min_{j\in\{1,\ldots,r\}} \left\{\frac{m_j}{s_j} \right\},
\]
where $s_j$ is the probability that by creating a new hyperedge a community $j$ is chosen as the one sharing it. 
\end{theorem}

\begin{remark}
	The value $s_j$ can be derived from $P$; it is the sum of probabilities of creating a hyperedge between $C^{(j)}$ and any other subset of communities. 
\end{remark}

	\section{Modularity of $\mathbf{G(G_0,p,M,X,P,\gamma)}$} \label{sec:modularity}
	
In this section we give lower bounds for the modularity of $G = G(G_0,p,M,X,P,\gamma)$ in terms of the values from matrix $P$. We present experimental results showing the advantage in modularity of our model over the one in~\cite{ALP_hyper}.




\subsection{Theoretical results}

We analyse $G(G_0,p,M,X,P,\gamma)=(V,E)$ obtained up to time $t$ (this time we omit superscripts $t$). Recall that each vertex from $V$ is assigned to one of $r$ communities, $V = C^{(1)} \mathbin{\dot{\cup}} C^{(2)} \mathbin{\dot{\cup}} \ldots \mathbin{\dot{\cup}} C^{(r)}$. We obtain the lower bound for modularity deriving the modularity score of the partition $\mathcal{C} = \{C^{(1)}, C^{(2)}, \ldots, C^{(r)}\}$. This choice of partition seems obvious provided that matrix $P$ is strongly assortative, i.e., the probabilities of having an edge inside communities are all bigger than the highest probability of having an edge joining different communities. Note that what matters for the value of modularity is the total sum of degrees in each community, not the distribution of degrees. Therefore we do not use the fact that the degree distribution follows a power-law in each community and in the whole model. 
We just use information from matrix $P$. Thus, in fact, we derive the lower bound for the modularity of stochastic block model with $r$ communities.

For $\ell \geqslant 1$ $E_{\ell} \subseteq E$ is the set of hyperedges of cardinality $\ell$. First, we state general lower bound for the modularity of $G$ ias a function of matrix $P$.

\begin{lemma} \label{lemma:mod_hyper_gen}
Let $G = G(G_0,p,M,X,P,\gamma)$ 
with the size of each hyperedge bounded by $d$. 
Let $p_i$ be the probability that a randomly chosen hyperedge is within community $C^{(i)}$ \textnormal{(}i.e., all vertices of a hyperedge belong to $C^{(i)}$\textnormal{)}. By $s_i$ we denote the probability that a randomly chosen hyperedge has at least one vertex in community $C^{(i)}$.  Assume also that with high probability $|E_{\ell}|/|E| \sim a_{\ell}$ for some constants $a_{\ell} \in [0,1]$ and $vol(V)/|E| \sim \delta$ for some constant $\delta \in (0,\infty)$. Then whp
\[
\lim_{t \rightarrow \infty} q^*(G) \geqslant \sum_{i=1}^{r} p_i - \sum_{i=1}^{r} \sum_{\ell \geqslant 1} a_{\ell}\left(\frac{(d-1)s_i+p_i}{\delta}\right)^{\ell}.
\]
\end{lemma}
\begin{remark}
Note that for $G$ being $2$-uniform (thus simply a graph) this result simplifies significantly to $\lim_{t \rightarrow \infty} q^*(G) \geqslant  \sum_{i=1}^{r} p_{i} - 1/4 \sum_{i=1}^{r}(s_i+p_{i})^2$.
\end{remark}

Below we state the lower bound for the modularity of $G$ in a version in which the knowledge of the whole matrix $P$ is not necessary. Instead we use its two characteristics: $\alpha$ - the probability that a randomly chosen hyperedge joins at least two different communities (may be interpreted as the amount of noise in the network) and $\beta$ - the maximum value among $p_{i}$'s for $i \in \{1,2,\ldots,r\}$. The modularity of the model will be maximised for $\alpha=0$ (when there are no hyperedges joining different communities) and $\beta = 1/r$ (when all $p_{i}$'s are equal to $1/r$ thus hyperedges are distributed uniformly across communities).

\begin{lemma} \label{lemma:mod_hyper_ab}
By assumptions from Lemma \ref{lemma:mod_hyper_gen} whp $ \lim_{t \rightarrow \infty} q^*(G) \geqslant$ \\ $1-\alpha - a_1 \left(\frac{d}{\delta}\right)((d-2)\alpha+1) - \sum_{\ell \geqslant 2} a_{\ell}\left(\frac{d}{\delta}\right)^{\ell}\left((r-1)\beta^{\ell} + ((d-1)\alpha + \beta)^{\ell}\right)$,\\
where $\alpha = 1 - \sum_{i=1}^{r}p_{i}$ and $\beta = \max_{i \in \{1,\ldots,r\}} p_{i}$.

\end{lemma}
\begin{remark}
	For $G$ being $2$-uniform, 
	the result simplifies to $\lim_{t \rightarrow \infty} q^*(G) \geqslant 1 - r \beta^2 - \alpha(1 + \alpha + 2\beta)$. Note that for $\alpha=0$ and $\beta=1/r$, this bound equals $1-1/r$ and is tight, i.e., it is the modularity of the graph with the same number of edges in each of its $r$ communities and no edges between different communities. 
\end{remark}
\begin{remark}
	Obtained bounds work well as long as the cardinalities of hyperedges do not differ too much. This is since deriving them we bound the cardinality of each hyperedge by the size of the biggest one. In particular, the bounds are very good in case of uniform hypergraphs - check experimental results below.
\end{remark}

\subsection{Experimental results} \label{sec:mod_exp}

In this subsection we show  how the modularity of our model $G$ compares with Avin et al. hypergraph $A$ \cite{ALP_hyper} and with a real-life co-authorship graph $R$. We also check how good is our theoretical lower bound for modularity.

\vspace{-0.2cm}
\begin{figure}[ht]
	\begin{minipage}[b]{0.475\textwidth}
		\includegraphics[width=0.9\textwidth]{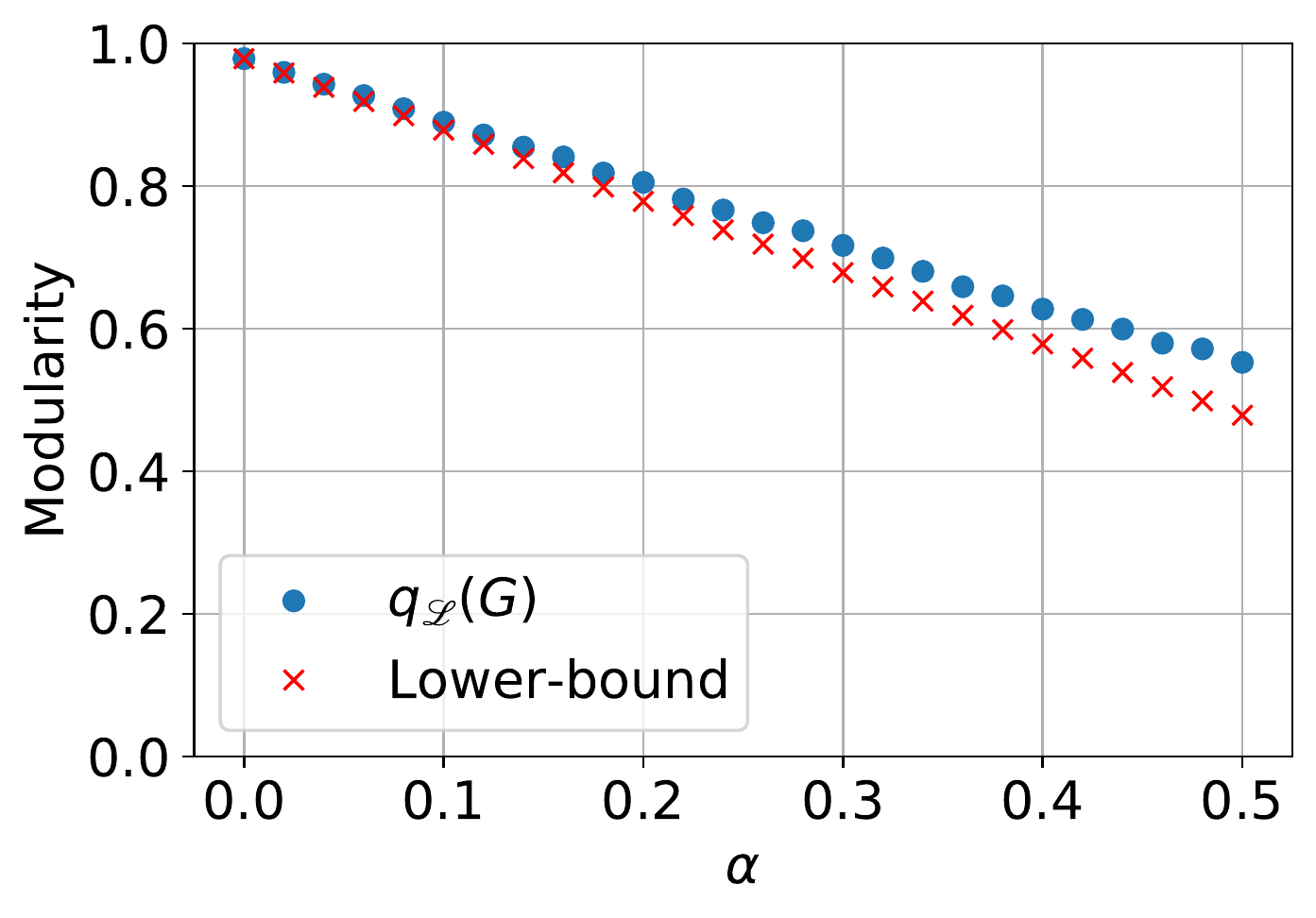}
		\vspace{4pt}
		\centering
		\vspace{-10pt}
		{\small $2$-uniform $G$}
	\end{minipage}
	\begin{minipage}[b]{0.475\textwidth}
		\includegraphics[width=0.9\textwidth]{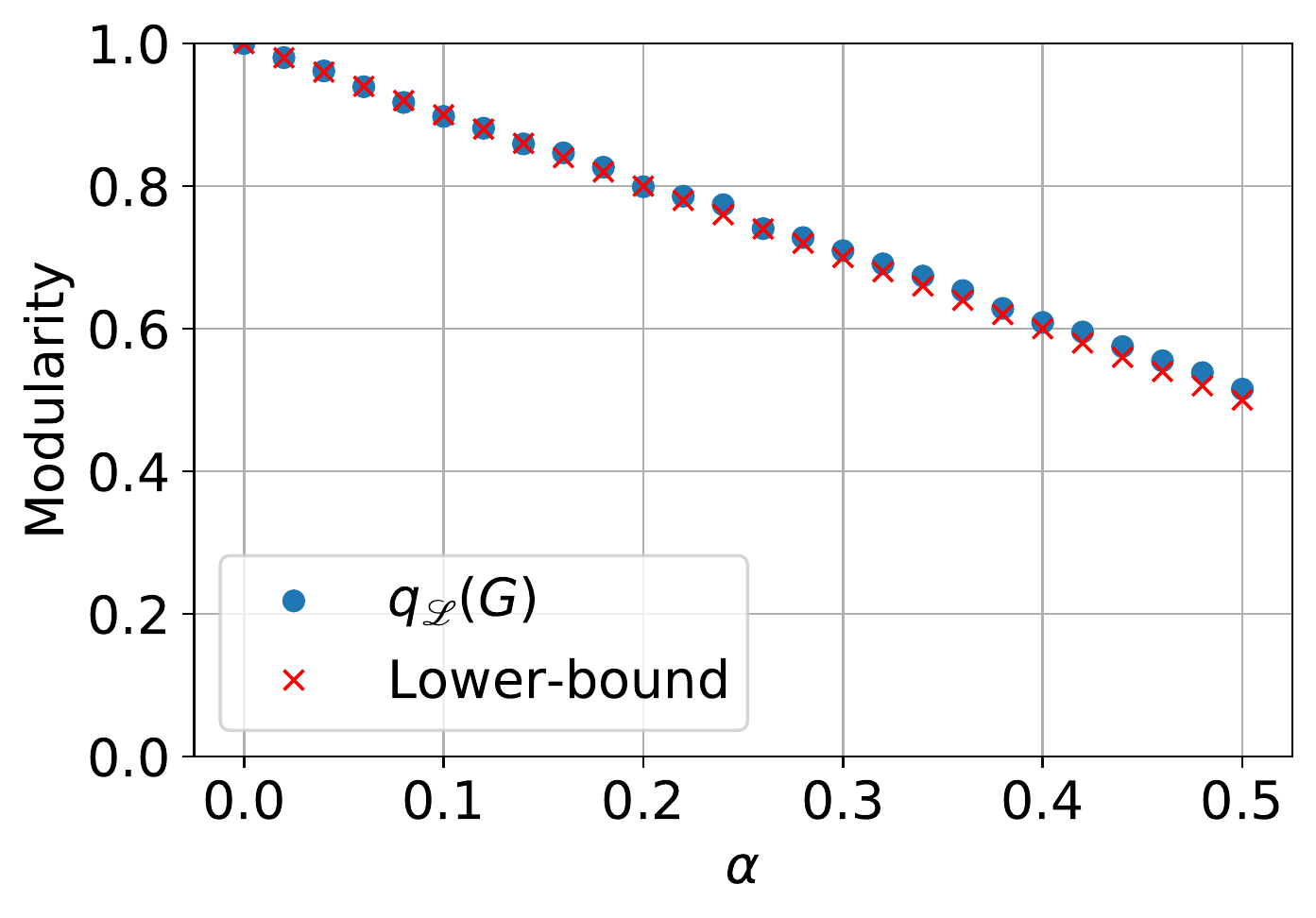}
		\vspace{4pt}
		\centering
		\vspace{-10pt}
		{\small $20$-uniform $G$}
	\end{minipage}
	\caption{Lower bound from Lemma \ref{lemma:mod_hyper_gen} in comparison with the modularity score obtained by Leiden algorithm on simulated uniform hypergraphs $G$.}
	\vspace{-15pt}
	\label{fig:lower_bound}
\end{figure}

 To get the approximation of modularity of simulated hypergraphs we used Leiden procedure \cite{Leiden} - a popular community detection algorithm for large networks. Calculating modularity is NP-hard \cite{modularity_NPhard}. Leiden is nowadays one of the best heuristics trying to find a partition maximising modularity. Therefore we treat its outcome partition as the one whose modularity score is quite precise approximation of the modularity of graphs in question. Every presented modularity score (using Definition \ref{def:mod_hypergraphs}) refers to a partition returned by Leiden algorithm ran on the flattened hypergraph (i.e., a graph obtained from a hypergraph by exchanging hyperedges with cliques). We did not manage to run Leiden-like algorithm directly for hypergraphs due to their big scale and our technical limitations.

\begin{wrapfigure}{R}{0.5\textwidth}
	\vspace{-25pt}
	\centering{
		\includegraphics[width = 0.45 \textwidth]{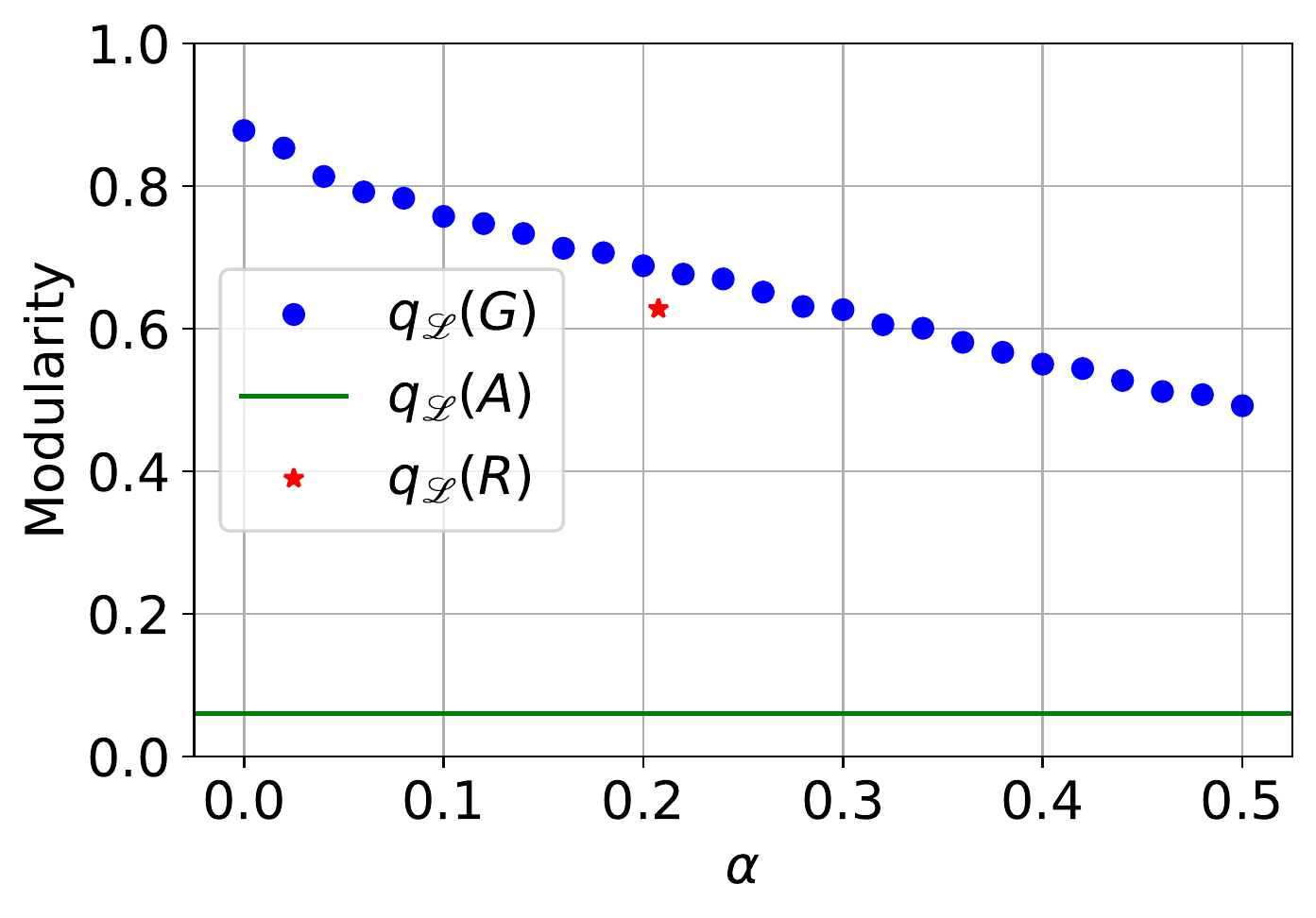}
		\vspace{-15pt}
		\caption{{\small Comparison of modularity between our model $G$, Avin et al. hypergraph $A$ and real co-authorship hypergraph $R$.}} \label{fig:mod_real}
	}
	\vspace{-20pt}
\end{wrapfigure}

Fig.~\ref{fig:lower_bound} shows the lower bound from Lem.~\ref{lemma:mod_hyper_gen} in comparison with the modularity of $2-$ and $20$-uniform hypergraph $G(G_0,p,M,X,P,\gamma)$ on $10^4$ vertices, where $M$ is uniform and matrix $P$ has values $(1-\alpha)/47$ ($47$ is the number of communities also in $R$) on the diagonal and the rest of probability mass spread uniformly over remaining entries.  
As we expected - the theoretical bound almost overlapped with the value of modularity in this case.

To build a real-life co-authorship hypergraph $R$ we used data downloaded from the citation database Scopus~\cite{scopus}. 
We have considered articles across all disciplines from the years 1990-2018 with at least one French co-author. Obtained hypergraph consisted of $\approx 2.2 \cdot 10^6$ nodes (authors) and $\approx 3.9 \cdot 10^6$ hyperedges (articles). Next, we implemented our model $G$ and Avin's et al. model $A$ using the parameters (distribution of hyperedges cardinalities, vector $M$, matrix $P$) gathered from hypergraph $R$. Figure~\ref{fig:mod_real} compares modularities of $G$, $A$, and $R$. For $R$ the value $\alpha$ equals $0.21$. Then the modularity of our model is around $0.69$ which is very close to the modularity of $R$ ($\approx 0.63$). The modularity of $A$, as $A$ does not feature communities, is very low ($\approx 0.06$). Figure~\ref{fig:mod_real} shows also how the modularity of $G$ changes with $\alpha$ and one may notice that it stays at reasonably high level even when the amount of the noise in the network grows.
 

\vspace{-0.1cm}

	\section{Conclusion and Further Work} \label{sec:summary}



We have proved theoretically and confirmed experimentally that our model exhibits high modularity, which is rare for known preferential attachment graphs and was not present in hypergraph models so far. While our model has many parameters and may seem complicated, this general formulation allowed us to unify many results known so far. Moreover, 
it can be easily transformed into much simpler model (e.g., by setting some arguments trivially to $0$, repeating the same distributions for hyperedges cardinalities...).

It is commonly known that many real networks present an exponential cut-off in their degree distribution. One possible reason to explain this phenomenon is that nodes eventually become inactive in the network. As further work, we will include 
this process in our model. The other direction of future study is making the preferential attachment depending not only on the degrees of the vertices but also on their own characteristic (generally called fitness). 
	
	%
	%
	\newpage
	\bibliographystyle{splncs04}
	\bibliography{hyper_mod}
	
	\newpage
	\section*{Appendix}
	
\subsection*{Degree distribution of  $\mathbf{H(H_0,p,Y,X,m,\gamma)}$}

The number of vertices in $H_t$ is a random variable following a binomial distribution. Since $|V_0|=1$
we have $|V_t| \sim B(t,p_v+p_{ve}) + 1$. Since $|E_0|=1$, the number of hyperedges in $H_t$ is a random variable satisfying $|E_t| \sim m  B(t,p_{ve}+p_e) + 1$.

Before we prove Theorem \ref{thm:hyper_plaw} we discuss briefly the concentration of random variables $|V_t|$ (the number of vertices at time $t$), $D_t$ (the sum of degrees at time $t$) and $W_t = D_{t} + \gamma |V_t|$. We also state two technical lemmas that will be helpful later on. 

\begin{lemma}[Chernoff bounds, \cite{MU_book}, Chapter 4.2] \label{lemma:Chernoff}
	Let $Z_1,Z_2,\ldots,Z_t$ be independent indicator random variables with $\Pr[Z_i=1]=p_i$ and $\Pr[Z_i=0]=1-p_i$. Let $\delta>0$, $Z = \sum_{i=1}^{t} Z_i$ and $\mu = \E[Z] = \sum_{i=1}^{t} p_i$. Then
	\[
	\Pr[|Z-\mu| \geqslant \delta\mu] \leq 2 e^{-\mu\delta^2/3}.
	\]
\end{lemma}

\begin{corollary} \label{cor:Vt_concentr}
	Since $|V_t| \sim B(t,p_v+p_{ve}) + 1$ setting $\delta = \sqrt{\frac{9 \ln{t}}{(p_v+p_{ve})t}}$ in Lemma \ref{lemma:Chernoff} 
	we get
	\[
	\Pr[||V_t|-\E[]|V_t|]| \geqslant \sqrt{9 (p_v+p_{ve}) t \ln{t}}] \leq 2/t^3.
	\]
\end{corollary}

\begin{lemma} \label{lemma:Nkt_limit}
	If $\lim_{t \rightarrow \infty} \frac{\E[N_{k,t}]}{t} \sim c k^{-\beta}$ for some positive constant $c$ then
	\[
	\lim_{t \rightarrow \infty} \E\left[\frac{N_{k,t}}{|V_t|}\right] \sim \frac{c}{p_v+p_{ve}} k^{-\beta}.
	\]
	\textnormal{(}Here ``$\sim$'' refers to the limit by $k \rightarrow \infty$.\textnormal{)}
\end{lemma}

\begin{proof}
	Let $(\Omega, {\mathcal{F}}, \Pr)$ be the probability space on which random variables $N_{k,t}$ and $|V_t|$ are defined. Thus $N_{k,t}:\Omega \rightarrow \mathbb{R}$ and $|V_t|:\Omega \rightarrow \mathbb{R}$. Let $\Omega_1 \subseteq \Omega$ denote the set of all $\omega \in \Omega$ such that $|V_t|(\omega) \in (\E|V_t|-\sqrt{9 (p_v+p_{ve}) t \ln{t}}, \E|V_t|+\sqrt{9 (p_v+p_{ve}) t \ln{t}})$. By Corollary \ref{cor:Vt_concentr} we know that $\sum_{\omega \in \Omega \setminus \Omega_1} \Pr[\omega] \leqslant 2/t^3$. Using the fact that for each $\omega$ $\frac{N_{k,t}(\omega)}{|V_t|(\omega)} \leqslant 1$ we get
	\[
	\begin{split}
		\E\left[\frac{N_{k,t}}{|V_t|}\right] & = \sum_{\omega \in \Omega} \frac{N_{k,t}(\omega)}{|V_t|(\omega)}\Pr[\omega] = \sum_{\omega \in \Omega_1} \frac{N_{k,t}(\omega)}{|V_t|(\omega)}\Pr[\omega] + \sum_{\omega \in \Omega\setminus\Omega_1} \frac{N_{k,t}(\omega)}{|V_t|(\omega)}\Pr[\omega] \\
		& \leqslant \sum_{\omega \in \Omega} \frac{N_{k,t}(\omega)}{\E|V_t|-\sqrt{9 (p_v+p_{ve}) t \ln{t}}}\Pr[\omega] + \sum_{\omega \in \Omega\setminus\Omega_1} 1 \cdot \Pr[\omega] \\
		& \leqslant \frac{\E[N_{k,t}]}{\E|V_t|-\sqrt{9 (p_v+p_{ve}) t \ln{t}}} + 2/t^3 \sim \frac{\E[N_{k,t}]}{(p_v + p_{ve})t}.
	\end{split}
	\]
	On the other hand, since $N_{k,t} \leq t$,
	\[
	\begin{split}
		\E\left[\frac{N_{k,t}}{|V_t|}\right] & \geqslant \sum_{\omega \in \Omega_1} \frac{N_{k,t}(\omega)}{|V_t|(\omega)}\Pr[\omega] \geqslant \sum_{\omega \in \Omega_1} \frac{N_{k,t}(\omega)}{\E|V_t|+\sqrt{9 (p_v+p_{ve}) t \ln{t}}}\Pr[\omega] \\
		& = \frac{1}{\E|V_t|+\sqrt{9 (p_v+p_{ve}) t \ln{t}}} \left(\E[N_{k,t}] - \sum_{\omega \in \Omega\setminus\Omega_1} N_{k,t}(\omega) \Pr[\omega]\right) \\
		& \geqslant \frac{1}{\E|V_t|+\sqrt{9 (p_v+p_{ve}) t \ln{t}}} \left(\E[N_{k,t}] - \sum_{\omega \in \Omega\setminus\Omega_1} t \cdot \Pr[\omega]\right) \\
		& \geqslant \frac{\E[N_{k,t}]}{\E|V_t|+\sqrt{9 (p_v+p_{ve}) t \ln{t}}} - \frac{t \cdot 2/t^3}{\E|V_t|+\sqrt{9 (p_v+p_{ve}) t \ln{t}}} \\
		& \sim \frac{\E[N_{k,t}]}{(p_v + p_{ve})t}.
	\end{split}
	\]\qed
\end{proof}

\begin{lemma} [Hoeffding's inequality, \cite{Hoeff}] \label{lemma:hoeff}
	Let $Z_1, Z_2, \ldots, Z_t$ be independent random variables such that $\Pr[Z_i \in [a_i,b_i]] =1$. Let $\delta>0$ and $Z = \sum_{i=1}^{t}Z_i$. Then
	\[
	\Pr[|Z - \E[Z]| \geqslant \delta] \leqslant 2 \exp\left\{-\frac{2 \delta^2}{\sum_{i=1}^t(a_i-b_i)^2}\right\}.
	\]
\end{lemma}

\begin{lemma} \label{lemma:Wt_concentr}
	Let $t>0$. Let $\E[Y_t] = \mu_0$, and $\E[X_t^{(i)}] = \mu_i$ for $i \in \{1,2,\ldots,r\}$. Moreover, let $2 \leqslant Y_t < t^{1/4}$ and $1 \leqslant X_t^{(i)} < t^{1/4}$ for $i \in \{1,2,\ldots,r\}$. Let $W_t = D_t + \gamma|V_t|$. Then
	\[
	\Pr[|W_t - \E[W_t]| \geq m t^{3/4}\sqrt{2\ln{t}}] = \mathcal{O}\left(\frac{1}{t^4}\right).
	\]
\end{lemma}

\begin{proof}
	Our initial hypergraph consists of a single hyperedge of cardinality $1$ over a single vertex thus $W_0 = \gamma + 1$. For $t \geq 1$ we can obtain $W_t$ from $W_{t-1}$ by adding:
	\begin{enumerate}
		\item either $\gamma$ with probability $p_v$,
		\item or $\gamma + m Y_t$ with probability $p_{ve}$,
		\item or $m X_t^{(1)}$ with probability $p_e^{(1)}$,
		\item or $m X_t^{(2)}$ with probability $p_e^{(2)}$,
		\item or $\ldots$,
		\item or $m X_t^{(r)}$ with probability $p_e^{(r)}$,
		\item or $0$ with probability $1-p_v-p_{ve}-p_e$.
	\end{enumerate}
	Thus we can express $W_t$ as the sum of independent random variables $W_t = \gamma + 1 + Z_1 + Z_2 + \ldots + Z_t$, where $\E[Z_i] = \gamma \bar{V} + \bar{D}$ and $1 \leqslant Z_i \leqslant m t^{1/4} + \gamma$ for $i \in \{1,2,\ldots,t\}$ and $\bar{D}$ and $\bar{V}$ are defined as in Theorem \ref{thm:hyper_plaw}:
	\[
	\bar{V}  = p_v+p_{ve} \quad \textnormal{and} \quad \bar{D} = m(p_{ve}\mu_0 + p_e^{(1)}\mu_1 + \ldots + p_e^{(r)}\mu_r).
	\]
	Now, setting $\delta = m t^{3/4}\sqrt{2 \ln{t}}$ in Hoeffding's inequality (see Lemma \ref{lemma:hoeff}) we get
	\[
	\Pr[|W_t - \E[W_t]| \geqslant m t^{3/4}\sqrt{2 \ln{t}}] \leqslant 2 \exp\left\{-\frac{4 \cdot m^2 \cdot t^{6/4} \cdot \ln{t}}{(t+1) (m \cdot t^{1/4} + \gamma)^2}\right\} = \mathcal{O}\left(\frac{1}{t^4}\right).
	\]\qed
\end{proof}

\begin{lemma}[\cite{ChLu_book}, Chapter 3.3] \label{lemma:rec_seq}
	Let $\{a_t\}$ be a sequence satisfying the recursive relation
	\[
	a_{t+1} = \left(1-\frac{b_t}{t}\right)a_t + c_t
	\]
	where $b_t \xrightarrow{t \rightarrow \infty} b>0$ and $c_t \xrightarrow{t \rightarrow \infty} c$. Then the limit $\lim_{t \rightarrow \infty} \frac{a_t}{t}$ exists and
	\[
	\lim_{t \rightarrow \infty} \frac{a_t}{t} = \frac{c}{1+b}.
	\]
\end{lemma}
\noindent
Now we are ready to prove Theorem \ref{thm:hyper_plaw}.

\begin{proof}[Theorem \ref{thm:hyper_plaw}]
	Here we take a standard master equation approach that can be found e.g. in Chung and Lu book \cite{ChLu_book} about complex networks or Avin et al. paper \cite{ALP_hyper} on preferential attachment hypergraphs.
	
	Recall that $N_{k,t}$ denotes the number of vertices of degree $k$ at time $t$. We need to show that
	\[
	\lim_{t \rightarrow \infty} \E\left[\frac{N_{k,t}}{|V_t|}\right] \sim \tilde{c} k^{-\beta}
	\]
	for some constant $\tilde{c}$ and $\beta = 2+\frac{\gamma \bar{V} + m \cdot p_{ve}}{\bar{D} - m \cdot p_{ve}}$. However, by Lemma \ref{lemma:Nkt_limit} we know that it suffices to show that
	\[
	\lim_{t \rightarrow \infty} \frac{\E[N_{k,t}]}{t} \sim c k^{-\beta}
	\]
	for some constant $c$.
	
	Our initial hypergraph $H_0$ consists of a single hyperedge of cardinality $1$ over a single vertex thus we can write $N_{0,0} = 0$ and $N_{1,0}=1$. Now, to formulate a recurrent master equation we make the following observation for $t \geqslant 1$. The vertex $v$ has degree $k$ at time $t$ if either it had degree $k$ at time $t-1$ and was not chosen to any new hyperedge or it had degree $k-i$ at time $t-1$ and was chosen $i$ times to new hyperedges. Note that $i$ can be at most $\min\{k,mZ_t\}$, where $Z_t$ represents a random variable chosen among $Y_t, X_t^{(1)}, \ldots, X_t^{(r)}$ according to $(p_v,p_{ve},p_e^{(1)},\ldots,p_e^{(r)})$. Additionally, at each time step a vertex of degree $0$ may appear as the new one with probability $p_v$ and a vertex of degree $m$ may appear as the new one with probability $p_{ve}$. 
	Let ${\mathcal{F}}_t$ be the $\sigma$-algebra associated with the probability space at time $t$. Let $Q_{d,k,t}$ denote the probability that a specific vertex of degree $k$ was chosen $d$ times to be included in new hyperedges at time $t$ (this probability is expressed as a random variable since it depends on a specific realization of the process up to time $t-1$). Let also $W_t = D_{t}+\gamma|V_{t}|$. For $t \geqslant 1$ we get
	\[
	\E[N_{0,t} |{\mathcal{F}}_{t-1}] = p_v + N_{0,t-1} Q_{0,0,t}
	\]
	and when $k \geqslant 1$
	\[
	\begin{split}
		\E[N_{k,t} |{\mathcal{F}}_{t-1}] & = \delta_{k,m} p_{ve} + N_{k,t-1} Q_{0,k,t} + N_{k-1,t-1} Q_{1,k-1,t} \\
		& \quad + \sum_{i=2}^{\min\{k,mZ_t\}} N_{k-i,t-1} Q_{i,k-i,t},
	\end{split}
	\]
	where $\delta_{k,m}$ is the Kronecker delta. We have extracted the first two terms in the above sum since below we prove that these are the dominating terms.
	Taking expectation on both sides we obtain
	\begin{equation} \label{eq:master1}
		\E[N_{0,t}] = p_v + \E[N_{0,t-1} Q_{0,0,t}]
	\end{equation}
	and for $k \geqslant 1$
	\begin{equation} \label{eq:master2}
		\begin{split}
			\E[N_{k,t}] & = \delta_{k,m} p_{ve} + \E[N_{k,t-1} Q_{0,k,t}] + \E[N_{k-1,t-1} Q_{1,k-1,t}] \\
			& \quad + \sum_{i=2}^{\min\{k,mZ_t\}} \E[N_{k-i,t-1} Q_{i,k-i,t}].
		\end{split}
	\end{equation}
	Note that
	\[
	\begin{split}
		Q_{0,k,t} & = p_v + (1-p_v-p_{ve}-p_e) + p_{ve} \E\left[\left(1-\frac{k+\gamma}{W_{t-1}}\right)^{m (Y_t-1)}|{\mathcal{F}}_{t-1}\right]\\
		&  \quad + p_e^{(1)} \E \left[\left(1-\frac{k+\gamma}{W_{t-1}}\right)^{m X_t^{(1)}}|{\mathcal{F}}_{t-1}\right] + \ldots \\
		& \quad + p_e^{(r)} \E \left[\left(1-\frac{k+\gamma}{W_{t-1}}\right)^{m X_t^{(r)}}|{\mathcal{F}}_{t-1}\right]
	\end{split}
	\]
	while for $i \in \{1,2,\ldots,k\}$
	\[
	\begin{split}
		Q_{i,k-i,t} & = p_{ve} \E\left[\binom{m(Y_t-1)}{i}\left(\frac{k-i+\gamma}{W_{t-1}}\right)^i\left(1-\frac{k-i+\gamma}{W_{t-1}}\right)^{m(Y_t-1)-i} |{\mathcal{F}}_{t-1}\right] \\
		& \quad + p_e^{(1)}\E\left[ \binom{m X_t^{(1)}}{i}\left(\frac{k-i+\gamma}{W_{t-1}}\right)^i\left(1-\frac{k-i+\gamma}{W_{t-1}}\right)^{m X_t^{(1)}-i} |{\mathcal{F}}_{t-1} \right] + \ldots \\
		& \quad + p_e^{(r)} \E\left[\binom{m X_t^{(r)}}{i}\left(\frac{k-i+\gamma}{W_{t-1}}\right)^i\left(1-\frac{k-i+\gamma}{W_{t-1}}\right)^{m X_t^{(r)}-i} |{\mathcal{F}}_{t-1}\right].
	\end{split}
	\]
	Now, for any random variable $Z_t$ with constant expectation $\mu$, independent of the $\sigma$-algebra ${\mathcal{F}}_{t-1}$, and such that $1 \leq Z_t < t^{1/4}$, by Bernoulli's inequality we have
	\begin{equation} \label{eq:Z_lower}
		\begin{split}
			\E\left[\left(1-\frac{k+\gamma}{W_{t-1}}\right)^{m Z_t}|{\mathcal{F}}_{t-1}\right] & \geqslant \E\left[\left(1-\frac{mZ_t(k+\gamma)}{W_{t-1}}\right)|{\mathcal{F}}_{t-1}\right] \\
			& = 1 - \frac{m\mu(k+\gamma)}{W_{t-1}}.
		\end{split}
	\end{equation}
	On the other hand (using the fact that for $x \in [0,1]$ and $n \in \mathbb{N}$ we have $(1-x)^n \leq \frac{1}{1+nx}$)
	\begin{equation} \label{eq:Z_upper}
		\begin{split}
			\E\left[\left(1-\frac{k+\gamma}{W_{t-1}}\right)^{m Z_t}|{\mathcal{F}}_{t-1}\right] & \leqslant \E\left[\frac{1}{1+\frac{mZ_t(k+\gamma)}{W_{t-1}}}|{\mathcal{F}}_{t-1}\right] \\
			& = \E\left[1-\frac{mZ_t(k+\gamma)}{W_{t-1}+mZ(k+\gamma)}|{\mathcal{F}}_{t-1}\right] \\
			& \leqslant \E\left[1-\frac{mZ_t(k+\gamma)}{W_{t-1}} + \frac{(mZ_t(k+\gamma))^2}{W_{t-1}^2}|{\mathcal{F}}_{t-1}\right] \\
			& \leqslant 1 - \frac{m\mu(k+\gamma)}{W_{t-1}} + \frac{t^{1/2}(m(k+\gamma))^2}{W_{t-1}^2},
		\end{split}
	\end{equation}
	where the last inequality follows from the assumption $Z_t < t^{1/4}$. Now, let us consider the master equation (\ref{eq:master2}) for $\E[N_{k,t}]$  term by term. We start with the expected number of vertices that had degree $k$ at time $t-1$ and are still of degree $k$ at time $t$. By (\ref{eq:Z_lower}), Lemma \ref{lemma:Wt_concentr} and the fact that $N_{k,t-1} \leq t$ we get 
	\[
	\begin{split}
		\E[N_{k,t-1} Q_{0,k,t}] & \geqslant \E\left[N_{k,t-1} \left(1 - \frac{(k+\gamma) m (p_{ve}(\mu_0-1) + p_e^{(1)}\mu_1 + \ldots + p_e^{(r)}\mu_r)}{W_{t-1}} \right)\right] \\
		& = \E\left[N_{k,t-1} \left(1 - \frac{(k+\gamma) (\bar{D}-mp_{ve})}{W_{t-1}} \right)\right] \\
		& \geqslant \E[N_{k,t-1}] \left(1 - \frac{(k+\gamma) (\bar{D}-mp_{ve})}{\E[W_{t-1}]-mt^{3/4}\sqrt{2 \ln{t}}}\right) - t \cdot \frac{1}{t^4}.
	\end{split}
	\]
	To get the last inequality one needs to conduct calculations analogous to those from the proof of Lemma \ref{lemma:Nkt_limit}. By \ref{eq:Z_upper} and additionally using the fact that $W_{t-1} \geq 1$
	\[
	\begin{split}
		\E[N_{k,t-1} Q_{0,k,t}] & \leqslant \E\left[N_{k,t-1} \left(1 - \frac{(k+\gamma) (\bar{D}-mp_{ve})}{W_{t-1}}  + \frac{t^{1/2}(p_{ve}+p_e)(m(k+\gamma))^2}{W_{t-1}^2}\right)\right] \\
		& \leqslant \E[N_{k,t-1}] \left(1-\frac{(k+\gamma) (\bar{D}-mp_{ve})}{\E[W_{t-1}]+mt^{3/4}\sqrt{2 \ln{t}}}+\frac{t^{1/2}(p_{ve}+p_e)(m(k+\gamma))^2}{(\E[W_{t-1}]-mt^{3/4}\sqrt{2 \ln{t}})^2} \right) \\
		& \quad + \left(t + t^{3/2}(p_{ve}+p_e)(m(k+\gamma))^2\right) \cdot \frac{1}{t^4}.
	\end{split}
	\]
	Again, for the last inequality, proceed as in the proof of Lemma \ref{lemma:Nkt_limit}.
	Since $\E[W_{t-1}] = \bar{D}(t-1) + \gamma\bar{V}(t-1)$ and $\E[N_{k,t-1}] \leqslant t$, we obtain for fixed $k$
	\begin{equation} \label{eq:equ1}
		\E[N_{k,t-1} Q_{0,k,t}] = \E[N_{k,t-1}] \left(1 - \frac{(k+\gamma) (\bar{D}-mp_{ve})}{t(\bar{D} + \gamma\bar{V})+\mathcal{O}(t^{3/4}\sqrt{\ln{t}})}\right) + {\mathcal{O}}\left(\frac{1}{\sqrt{t}}\right).
	\end{equation}
	We treat $\E[N_{k-1,t-1} Q_{1,k-1,t}]$ similarly. On one hand we have
	\[
	\begin{split}
		Q_{1,k-1,t} & \geqslant  p_{ve} \E\left[m(Y_t-1)\frac{k-1+\gamma}{W_{t-1}}\left(1-\frac{mY_t(k-1+\gamma)}{W_{t-1}} \right)|{\mathcal{F}}_{t-1}\right] \\
		& \quad + p_{e}^{(1)} \E\left[mX_t^{(1)}\frac{k-1+\gamma}{W_{t-1}}\left(1-\frac{mX_t^{(1)}(k-1+\gamma)}{W_{t-1}} \right)|{\mathcal{F}}_{t-1}\right] + \ldots \\
		& \quad + p_{e}^{(r)} \E\left[mX_t^{(r)}\frac{k-1+\gamma}{W_{t-1}}\left(1-\frac{mX_t^{(r)}(k-1+\gamma)}{W_{t-1}} \right)|{\mathcal{F}}_{t-1}\right] \\
		& \geqslant p_{ve} \E\left[m(Y_t-1)\frac{k-1+\gamma}{W_{t-1}}|{\mathcal{F}}_{t-1}\right] - p_{ve}\E\left[\frac{Y_t^2(m(k-1+\gamma))^2}{W_{t-1}^2} |{\mathcal{F}}_{t-1}\right] + \ldots \\
		& \quad + p_{e}^{(r)} \E\left[m(X_t^{(r)})\frac{k-1+\gamma}{W_{t-1}}|{\mathcal{F}}_{t-1}\right] - p_{e}^{(r)}\E\left[\frac{(X_t^{(r)})^2(m(k-1+\gamma))^2}{W_{t-1}^2} |{\mathcal{F}}_{t-1}\right] \\
		& \geqslant \frac{p_{ve}m(\mu_0-1)(k-1+\gamma)}{W_{t-1}} - \frac{t^{1/2}p_{ve}(m(k-1+\gamma))^2}{W_{t-1}^2} + \ldots \\
		& \quad + \frac{p_{e}^{(r)}m\mu_r(k-1+\gamma)}{W_{t-1}} - \frac{t^{1/2}p_{e}^{(r)}(m(k-1+\gamma))^2}{W_{t-1}^2} \\
		& = \frac{(k-1+\gamma)(\bar{D}-mp_{ve})}{W_{t-1}} - \frac{t^{1/2}(p_{ve}+p_e)(m(k-1+\gamma))^2}{W_{t-1}^2}
	\end{split}
	\]
	(the last inequality follows from assumptions $Y_t < t^{1/4}$ and $X_t^{(i)}<t^{1/4}$), while on the other
	\[
	\begin{split}
		Q_{1,k-1,t} & \leqslant p_{ve} \E\left[m(Y_t-1)\frac{k-1+\gamma}{W_{t-1}}|{\mathcal{F}}_{t-1}\right] + \ldots + p_{e}^{(r)} \E\left[mX_t^{(r)}\frac{k-1+\gamma}{W_{t-1}}|{\mathcal{F}}_{t-1}\right] \\
		& \leqslant \frac{(k-1+\gamma)(\bar{D}-mp_{ve})}{W_{t-1}}. 
	\end{split}
	\]
	Again, by Lemma \ref{lemma:Wt_concentr}, the fact that $N_{k-1,t-1} \leqslant t$ and $N_{k-1,t-1}/W_{t-1}\leqslant 1$ for fixed $k$ we get 
	\begin{equation} \label{eq:equ2}
		\E[N_{k-1,t-1} Q_{1,k-1,t}] = \E[N_{k-1,t-1}] \left(\frac{(k-1+\gamma) (\bar{D}-mp_{ve})}{t(\bar{D} + \gamma\bar{V})+\mathcal{O}(t^{3/4}\sqrt{\ln{t}})}\right) + {\mathcal{O}\left(\frac{1}{\sqrt{t}}\right)}.
	\end{equation}
	The terms from equations (\ref{eq:equ1}) and (\ref{eq:equ2}) are those dominating in master equation (\ref{eq:master2}). For the sum of other terms we have the following upper bound when $k$ is fixed (the fourth inequality follows from upper bounding the sums by infinite geometric series and the asymptotics in the last line follows from Lemma \ref{lemma:Wt_concentr})
	\begin{equation} \label{eq:equ3}
		\begin{split}
			\sum_{i=2}^{\min\{k,m Z_t\}} & \E[N_{k-i,t-1}Q_{i,k-i,t}] \leqslant t \cdot \sum_{i=2}^k \E[Q_{i,k-i,t}] \\
			& \leqslant t \cdot \E\left[ \sum_{i=2}^k \left( p_{ve} \E\left[\binom{m(Y_t-1)}{i}\left(\frac{k-i+\gamma}{W_{t-1}}\right)^i|{\mathcal{F}}_{t-1}\right] \right.\right.\\
			& \quad \quad \quad \quad \quad \quad + p_e^{(1)}\E\left[ \binom{m X_t^{(1)}}{i}\left(\frac{k-i+\gamma}{W_{t-1}}\right)^i |{\mathcal{F}}_{t-1} \right] + \ldots \\
			& \left. \left. \quad \quad \quad \quad \quad \quad + p_e^{(r)} \E\left[\binom{m X_t^{(r)}}{i}\left(\frac{k-i+\gamma}{W_{t-1}}\right)^i |{\mathcal{F}}_{t-1}\right] \right)\right] \\
			& \leqslant t \cdot \E\left[ \E\left[ \sum_{i=2}^k \left( p_{ve} (mY_t)^i\left(\frac{k+\gamma}{W_{t-1}}\right)^i + \ldots \right.\right.\right.\\
			& \left. \left. \left. \quad \quad \quad \quad \quad \quad + p_e^{(r)}(m X_t^{(r)})^i\left(\frac{k+\gamma}{W_{t-1}}\right)^i  \right) | {\mathcal{F}}_{t-1}\right] \right] \\
			& \leqslant t \cdot \E\left[ \E\left[p_{ve} \frac{(m(k+\gamma)Y_t)^2}{W_{t-1}^2}\frac{1}{1-\frac{m(k+\gamma)Y_t}{W_{t-1}}} + \ldots \right. \right. \\
			& \left. \left. \quad \quad \quad \quad \quad \quad + p_{e}^{(r)} \frac{(m(k+\gamma)X_t^{(r)})^2}{W_{t-1}^2}\frac{1}{1-\frac{m(k+\gamma)X_t^{(r)}}{W_{t-1}}} | {\mathcal{F}}_{t-1}\right] \right] \\
			& \leqslant t \cdot \E \left[p_{ve} \frac{(m(k+\gamma)t^{1/4})^2}{W_{t-1}^2}\frac{1}{1-\frac{m(k+\gamma)t^{1/4}}{W_{t-1}}} + \ldots \right. \\
			& \left. \quad \quad \quad \quad \quad \quad + p_{e}^{(r)} \frac{(m(k+\gamma)t^{1/4})^2}{W_{t-1}^2}\frac{1}{1-\frac{m(k+\gamma)t^{1/4}}{W_{t-1}}}\right] \\
			& =  t \cdot \E \left[\frac{(p_{ve}+p_e)(m(k+\gamma))^2 t^{1/2}}{W_{t-1}^2}\frac{W_{t-1}}{W_{t-1}-m(k+\gamma)t^{1/4}} \right] \\
			& = (p_{ve}+p_e)(m(k+\gamma))^2 t^{3/2} \cdot \E\left[\frac{1}{W_{t-1}(W_{t-1}-m(k+\gamma)t^{1/4})}\right] \\
			& \sim (p_{ve}+p_e)(m(k+\gamma))^2 t^{3/2} \cdot \frac{1}{t^2} = \mathcal{O}\left(\frac{1}{\sqrt{t}}\right).
		\end{split}
	\end{equation}
	
	Plugging \ref{eq:equ1}, \ref{eq:equ2} and \ref{eq:equ3} into master equation (\ref{eq:master1}) and (\ref{eq:master2}) we obtain
	\begin{equation} \label{eq:N0} 
		\E[N_{0,t}] = \E[N_{0,t-1}] \left(1 - \frac{\gamma (\bar{D}-m p_{ve})}{t(\bar{D}+\gamma\bar{V}) + \mathcal{O}(t^{3/4}\sqrt{\ln{t}})} \right) + p_v + \mathcal{O}\left(\frac{1}{\sqrt{t}}\right)
	\end{equation}
	and
	\begin{equation} \label{eq:Nk}
		\begin{split}
			\E[N_{k,t}] & = \E[N_{k,t-1}] \left(1 - \frac{(k+\gamma) (\bar{D}-mp_{ve})}{t(\bar{D} + \gamma\bar{V})+\mathcal{O}(t^{3/4}\sqrt{\ln{t}})}\right)\\
			& \quad + \E[N_{k-1,t-1}] \left(\frac{(k-1+\gamma) (\bar{D}-mp_{ve})}{t(\bar{D} + \gamma\bar{V})+\mathcal{O}(t^{3/4}\sqrt{\ln{t}})}\right) + \delta_{k,m}p_{ve} + \mathcal{O}\left(\frac{1}{\sqrt{t}}\right).
		\end{split}
	\end{equation}

	For $k \geq 0$ by $L_k$ denote the limit
	\[
	L_k = \lim_{t \rightarrow \infty} \frac{\E[N_{k,t}]}{t}.
	\]
	First we prove that the limit $L_0$ exists. We apply Lemma \ref{lemma:rec_seq} to equation (\ref{eq:N0}) by setting 
	\[
	b_t = \frac{\gamma(\bar{D}-m p_{ve})}{\bar{D}+\gamma\bar{V} + \mathcal{O}(t^{3/4}\sqrt{\ln{t}}/t)} \quad \textnormal{and} \quad c_t = p_v + \mathcal{O}\left(\frac{1}{\sqrt{t}}\right).
	\]
	We get
	\[
	\lim_{t \rightarrow \infty} b_t = \frac{\gamma (\bar{D}-m p_{ve})}{\bar{D}+\gamma\bar{V}} \quad \quad \quad \textnormal{and} \quad \quad \quad \lim_{t \rightarrow \infty} c_t = p_v,
	\]
	therefore
	\[
	L_0 = \frac{p_v}{1+\frac{\gamma (\bar{D}-m p_{ve})}{\bar{D}+\gamma\bar{V}}} = \frac{p_v \frac{\bar{D}+\gamma\bar{V}}{\bar{D}-m p_{ve}}}{\frac{\bar{D}+\gamma\bar{V}}{\bar{D}-m p_{ve}} + \gamma}.
	\]
	Now, we assume that the limit $L_{k-1}$ exists and we will show by induction on $k$ that $L_k$ exists. Again, applying Lemma \ref{lemma:rec_seq} to equation (\ref{eq:Nk}) with
	\[
	b_t = \frac{(k+\gamma) (\bar{D}-m p_{ve})}{\bar{D}+\gamma\bar{V} + \mathcal{O}(t^{3/4}\sqrt{\ln{t}}/t)}
	\]
	and
	\[
	c_t = \frac{\E[N_{k-1,t-1}]}{t} \left(\frac{(k-1+\gamma) (\bar{D}-m p_{ve})}{\bar{D}+\gamma\bar{V} + \mathcal{O}(t^{3/4}\sqrt{\ln{t}}/t)}\right) + \delta_{k,m}p_{ve} + {\mathcal{O}}\left(\frac{1}{\sqrt{t}}\right)
	\]
	we get
	\[
	\lim_{t \rightarrow \infty} b_t = \frac{(k+\gamma) (\bar{D}-m p_{ve})}{\bar{D}+\gamma\bar{V}}
	\]
	and
	\[
	\lim_{t \rightarrow \infty} c_t = L_{k-1} \frac{(k-1+\gamma) (\bar{D}-m p_{ve})}{\bar{D}+\gamma\bar{V}} + \delta_{k,m}p_{ve},
	\]
	therefore
	\begin{equation} \label{eq:Nk_proportion}
		L_k = \frac{L_{k-1} \frac{(k-1+\gamma) (\bar{D}-m p_{ve})}{\bar{D}+\gamma\bar{V}} + \delta_{k,m}p_{ve}}{1+\frac{(k+\gamma) (\bar{D}-m p_{ve})}{\bar{D}+\gamma\bar{V}}} = \frac{L_{k-1}(k-1+\gamma)+ \delta_{k,m}p_{ve}\frac{\bar{D}+\gamma\bar{V}}{\bar{D}-m p_{ve}}}{k+\gamma+\frac{\bar{D}+\gamma\bar{V}}{\bar{D}-m p_{ve}}}.
	\end{equation}
	From now on, for simplicity of notation, we put $D = \frac{\bar{D}+\gamma\bar{V}}{\bar{D}-m p_{ve}}$ thus we have
	\[
	L_0 = \frac{p_v D}{\gamma + D} \quad \quad \textnormal{and} \quad \quad L_k = \frac{L_{k-1}(k-1+\gamma)+ \delta_{k,m}p_{ve}D}{k + \gamma + D}.
	\]
	When $k \in \{1,2,\ldots,m-1\}$, iterating over $k$ gives
	\[
	L_k = L_0 \cdot \prod_{\ell=1}^{k} \frac{\ell-1+\gamma}{\ell+\gamma+D} = \frac{p_v D}{\gamma+D} \prod_{\ell=1}^{k} \frac{\ell-1+\gamma}{\ell+\gamma+D}
	\]
	and when $k \geqslant m$
	\[
	\begin{split}
		L_k & = \frac{p_v D}{\gamma+D} \left( \prod_{\ell=1}^{k} \frac{\ell-1+\gamma}{\ell+\gamma+D} \right) + \frac{p_{ve} D}{m + \gamma + D} \left( \prod_{\ell=m+1}^{k} \frac{\ell-1+\gamma}{\ell+\gamma+D} \right)\\
		& = \left(\frac{p_v D}{\gamma+D} \left( \prod_{\ell=1}^{m} \frac{\ell-1+\gamma}{\ell+\gamma+D}\right) + \frac{p_{ve}D}{m + \gamma + D} \right) \left( \prod_{\ell=m+1}^{k} \frac{\ell-1+\gamma}{\ell+\gamma+D} \right)\\
		& = \left(\frac{p_v D}{\gamma+D} \frac{\Gamma(m+\gamma)}{\Gamma(\gamma)}\frac{\Gamma(\gamma + D +1)}{\Gamma(m+\gamma+D+1)} +\frac{p_{ve}D}{m + \gamma + D}\right) \\
		& \quad \quad \cdot \frac{\Gamma(m+\gamma+D+1)}{\Gamma(m+\gamma)}\frac{\Gamma(k+\gamma)}{\Gamma(k+\gamma+D+1)},
	\end{split}
	\]
	where $\Gamma(x)$ is the gamma function. Since $\lim_{k \rightarrow \infty} \frac{\Gamma(k)k^{\alpha}}{\Gamma(k+\alpha)} = 1$ for constant $\alpha \in \mathbb{R}$, we get
	\[
	\lim_{t \rightarrow \infty} \frac{\E[N_{k,t}]}{t} = L_k \sim c \cdot k^{-(1+D)}
	\]
	(``$\sim$'' refers to the limit by $k \rightarrow \infty$) for 
	\[
	c = p_v D\cdot\frac{\Gamma(\gamma+D)}{\Gamma(\gamma)} + p_{ve} D\cdot\frac{\Gamma(m+\gamma+D)}{\Gamma(m+\gamma)}.
	\]
	Hence, by Lemma \ref{lemma:Nkt_limit}, we obtain
	\[
	\lim_{t \rightarrow \infty} \E\left[\frac{N_{k,t}}{|V_t|}\right] \sim \frac{c}{p_v+p_{ve}} k^{-(1+D)}.
	\]
	We infer that the degree distribution of our hypergraph follows power-law with
	\[
	\beta = 1+D = 1+\frac{\bar{D}+\gamma\bar{V}}{\bar{D}-m p_{ve}} = 2 + \frac{\gamma\bar{V} + m p_{ve}}{\bar{D} - m p_{ve}}.
	\]\qed
\end{proof}

\subsection*{Degree distribution of $\mathbf{G(G_0,p,M,X,P,\gamma)}$}

The number of vertices in $G_t$ is a random variable satisfying $|V_t| \sim B(t,p) + r$, while for the number of hyperedges in $G_t$ we have $|E_t| \sim B(t,1-p) + r$. Note that since $|V_t|$ follows a binomial distribution, Lemma \ref{lemma:Nkt_limit} holds also in case of $G_t$ if we replace $p_v + p_{ve}$ by $p$.

Recall that $N_{k,t}$ stands for the number of vertices in $G_t$ of degree $k$. For $i \in \{1,2,\ldots,r\}$ by $N_{k,t}^{(i)}$ we denote the number of vertices of degree $k$ in $G_t$ belonging to community $C_t^{(i)}$. Thus $N_{k,t} = \sum_{i=1}^{r} N_{k,t}^{(i)}$.

\begin{lemma} \label{lemma:single_com}
	Consider a single community $C_{t}^{(j)}$ of a hypergraph $G_t$. Let $\E[X_t^{(i)}] = \mu_i$ and $1 \leqslant X_t^{(i)} < t^{1/4}$ for $i \in \{0,1,\ldots,r\}$. Then the degree distribution of vertices from $C_{t}^{(j)}$ follow a power-law with
	\[
	\beta_j = 2 + \frac{\gamma \bar{V}_j}{\bar{D}_j}
	\]
	where $\bar{V}_j$ is the expected number of vertices added to $C_{t}^{(j)}$ at a single time step and $\bar{D}_j$ is the average number of vertices from $C_{t}^{(j)}$ that increase their degree at a single time step, thus $\bar{V}_j = p m_j$ and $\bar{D}_j = (1-p) s_j\frac{\mu_1 + \ldots + \mu_r}{r}$, where $s_j$ is the probability that by creating a new hyperedge a community $j$ is chosen as the one sharing it \textnormal{(}we obtain $s_j$ from matrix $P$ - see remark below\textnormal{)}.
\end{lemma}

\begin{proof}
	Note that the community $C_{t+1}^{(j)}$ arises from community $C_{t}^{(j)}$ choosing at time $t$ only one of the following events according to $p$, $M$ and $P$.
	\begin{itemize}
		\item With probability $p m_j$: Add one new isolated vertex.
		\item With probability $\frac{(1-p)s_j}{r}$: Select $X_t^{(1)}$ vertices from $C_{t}^{(j)}$ in proportion to their degrees; these are vertices included in a newly created hyperedge, thus their degrees will increase.
		\item \ldots
		\item With probability $\frac{(1-p)s_j}{r}$: Select $X_t^{(r)}$ vertices from $C_{t}^{(j)}$ in proportion to their degrees; these are vertices included in a newly created hyperedge, thus their degrees will increase.
		\item With probability $1-(p m_j + (1-p)s_j)$: Do nothing.
	\end{itemize}
	Now, apply Theorem \ref{thm:hyper_plaw} with $p_v = p m_j$, $p_{ve}=0$, $p_e^{(1)} = p_e^{(2)} = \ldots = p_e^{(r)} = \frac{(1-p)s_j}{r}$ and $m=1$. We get that the degree distribution of vertices from $C_{t}^{(j)}$ follow a power-law with
	\[
	\beta_j = 2 + \frac{\gamma \bar{V}_j}{\bar{D}_j} = 2 + \frac{\gamma p m_j}{(1-p)s_j \frac{\mu_1 + \ldots + \mu_r}{r}}.
	\]\qed
\end{proof}

\begin{proof}[Theorem \ref{thm:degreed_G}]
	We need to prove that $\lim_{t \rightarrow \infty} \E\left[\frac{N_{k,t}}{|V_t|}\right] \sim \tilde{c} k^{-\beta}$ for some constant $\tilde{c}$ and $\beta$ as in the statement of theorem. By Lemma \ref{lemma:Nkt_limit} we know that it suffices to show $\lim_{t \rightarrow \infty} \frac{\E[N_{k,t}]}{t} \sim c k^{-\beta}$
	for some constant $c$. By Lemma \ref{lemma:Nkt_limit} and Lemma \ref{lemma:single_com} we write
	\[
	\begin{split}
		\lim_{t \rightarrow \infty} \frac{\E[N_{k,t}]}{t} & = \lim_{t \rightarrow \infty} \frac{\E[N_{k,t}^{(1)}]}{t} + \lim_{t \rightarrow \infty} \frac{\E[N_{k,t}^{(2)}]}{t} + \ldots + \lim_{t \rightarrow \infty} \frac{\E[N_{k,t}^{(r)}]}{t} \\
		& \sim c_1 k^{-\beta_1} + c_2 k^{-\beta_2} + \ldots + c_r k^{-\beta_r}
	\end{split}
	\]
	for some constants $c_1, \ldots, c_r$ and $\beta_j = 2 + \frac{\gamma \bar{V}_j}{\bar{D}_j}$. Thus $\lim_{t \rightarrow \infty} \frac{\E[N_{k,t}]}{t} \sim c k^{-\beta}$, where
	\[
	\beta  = \min_{j \in \{1,\ldots,r\}} \left\{\beta_j\right\}= 2 + \gamma \cdot \min_{j\in\{1,\ldots,r\}} \left\{ \frac{\bar{V}_j}{\bar{D}_j} \right\} = 2 + \frac{\gamma p }{(1-p)\frac{\mu_1 + \ldots + \mu_r}{r}} \cdot \min_{j\in\{1,\ldots,r\}} \left\{\frac{m_j}{s_j} \right\}.
	\]\qed
\end{proof}

In Figure \ref{fig:dd_real} we present log-log plots of a power-law distribution fitted to the degree distribution (DD) of a real-life co-authorship hypergraph $R$. $R$ is the same as in Section \ref{sec:modularity}. Thus it is built upon Scopus database, consists of $\approx 2.2 \cdot 10^6$ nodes (authors) and $\approx 3.9 \cdot 10^6$ hyperedges (articles). Left chart presents the degree distribution of the whole $R$ while the right one refers only to the biggest community of $R$ found by Leiden algorithm. One can observe the power-law behaviour in both cases.

\begin{figure}[ht]
	\begin{minipage}[b]{0.475\textwidth}
		\includegraphics[width=0.9\textwidth]{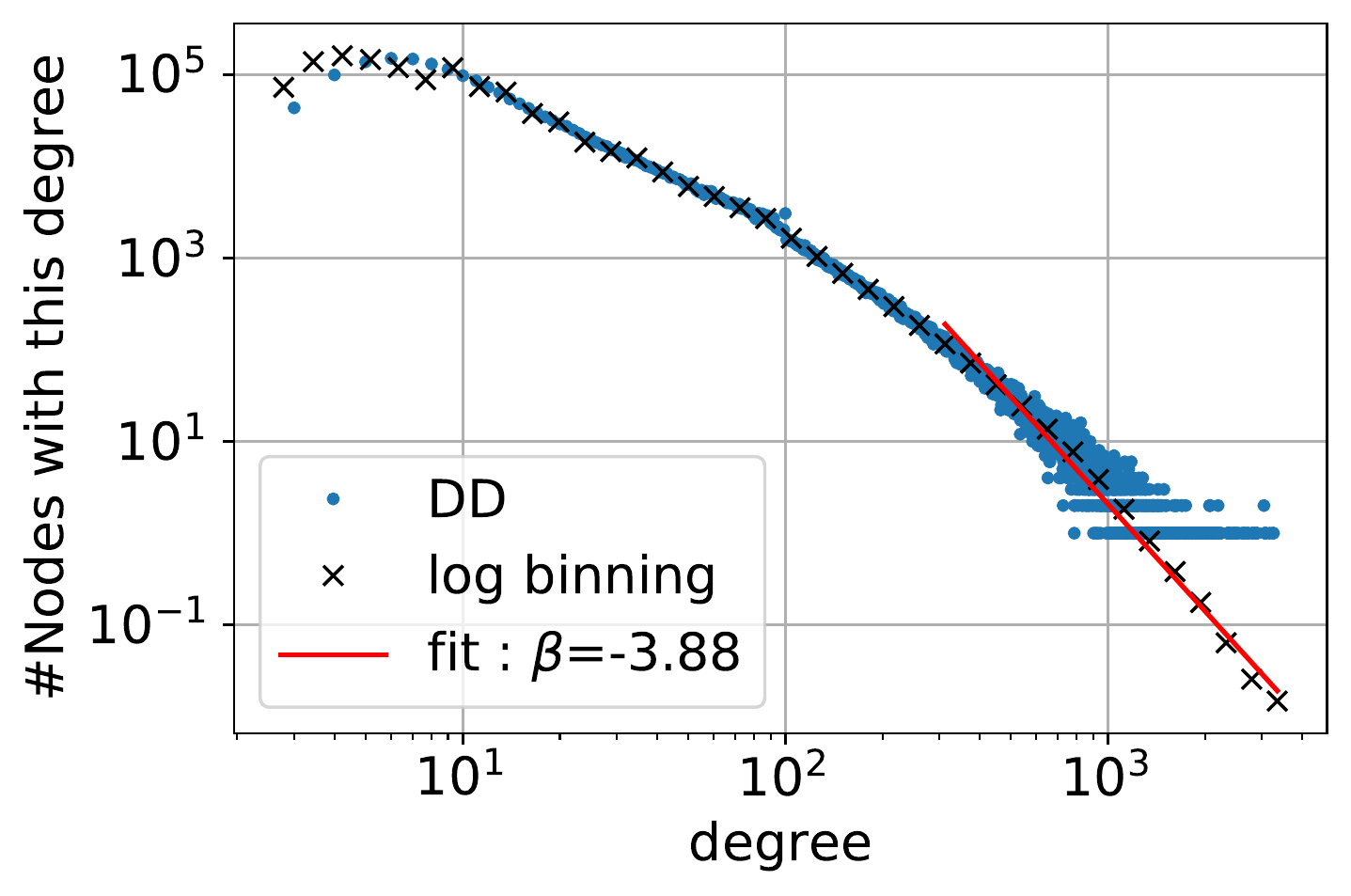}
		\vspace{4pt}
		\centering
		\vspace{-5pt}
		{\small Graph $R$}
	\end{minipage}
	\begin{minipage}[b]{0.475\textwidth}
		\includegraphics[width=0.9\textwidth]{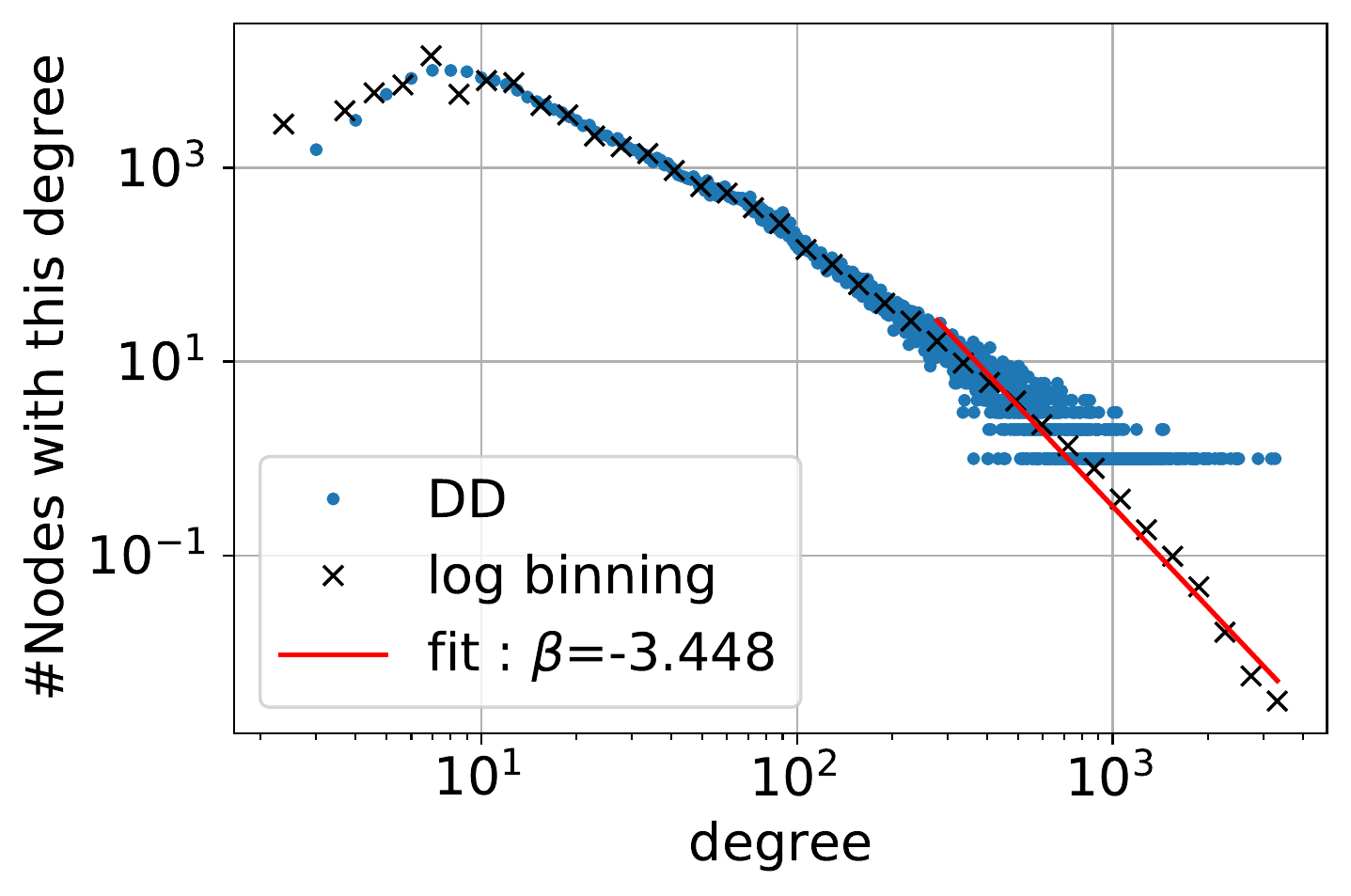}
		\vspace{4pt}
		\centering
		\vspace{-5pt}
		{\small The biggest community of $R$}
	\end{minipage}
	\caption{A power-law distribution fitted to the degree distributions.}
	\vspace{-15pt}
	\label{fig:dd_real}
\end{figure}

Let us also make one remark about the implementation of matrix $P$. 
\begin{remark}
	Observe that storing hyperedge probabilities in $d$-dimensional matrix $P$ we use much more space than we actually should. The same probabilities may repeat many times in $P$. E.g., when $d=2$ we get $2$-dimensional symmetric matrix $P$ such that $\sum_{i=1}^{r}\sum_{j=1}^i p_{ij}= 1$ and the the probability of creating hyperedge between two distinct communities $C^{(i)}$ and $C^{(j)}$ is in matrix $P$ doubled - as $p_{ij}$ and $p_{ji}$. If we allow for bigger hyperedges it may be repeated much more times. In fact we need to store at most $2^r-1$ different probabilities (one for each nonempty subset of the set of communities) while in $P$ we store $d^r$ values (in particular, if $d=r$ we store $r^r$ instead $2^r-1$ values). Nevertheless, for formal proofs this notation is convenient thus we use it at the same time underlining that implementation may be done much more space efficiently.
\end{remark}

\subsection*{Modularity}
\begin{proof}[Lemma \ref{lemma:mod_hyper_gen}]
	Let $\mathcal{C} = \{C^{(1)}, C^{(2)}, \ldots, C^{(r)}\}$.  Let also $q$ denote the probability of adding a new hyperedge in a single time step (hence $q=1-p$, referring to notation from Section \ref{sec:sbm}). Thus with high probability $|E| \sim t \cdot q$ (where `$\sim$' refers to the limit by $t\rightarrow \infty$). By Definition \ref{def:mod_hypergraphs} we write
	\[
	q^*(G) = \max_{\mathcal{A}} q_{\mathcal{A}}(G) \geqslant q_{\mathcal{C}}(G) = \sum_{i=1}^{r}\left(\frac{|E(C^{(i)})|}{|E|} - \sum_{\ell \geqslant 1}\frac{|E_{\ell}|}{|E|}\left(\frac{vol(C^{(i)})}{vol(V)}\right)^{\ell} \right).
	\]
	We obtain that with high probability
	\[
	q_{\mathcal{C}}(G) \sim \sum_{i=1}^{r}\left(\frac{t \cdot q \cdot p_{i}}{t \cdot q} - \sum_{\ell \geqslant 1}a_{\ell}\left(\frac{vol(C^{(i)})}{t \cdot q \cdot \delta}\right)^{\ell} \right).
	\]
	Note that if at a certain time step appears a hyperedge with all vertices contained in $C^{(i)}$, which happens with probability $q \cdot p_{i}$, it adds up at most $d$ to $vol(C^{(i)})$. If at a certain time step appears a hyperedge joining at least $2$ communities with at least one vertex in $C^{(i)}$, which happens with probability $q (s_i-p_{i})$, it adds up at most $d-1$ to $vol(C^{(i)})$. Thus we get that with high probability
	\[
	\begin{split}
		\lim_{t \rightarrow \infty} q^*(G) & \geqslant \sum_{i=1}^{r} p_i - \sum_{i=1}^{r} \sum_{\ell \geqslant 1} a_{\ell}\left(\frac{t \cdot q \cdot (d p_i + (d-1)(s_i-p_i))}{t \cdot q \cdot \delta}\right)^{\ell}\\
		& = \sum_{i=1}^{r} p_i - \sum_{i=1}^{r} \sum_{\ell \geqslant 1} a_{\ell}\left(\frac{(d-1)s_i+p_i}{\delta}\right)^{\ell}.
	\end{split}
	\]\qed
\end{proof}

\begin{proof}[Lemma \ref{lemma:mod_hyper_ab}]
	Let $\mathcal{C} = \{C^{(1)}, C^{(2)}, \ldots, C^{(r)}\}$ and for $i \in \{1,2,\ldots,r\}$ let $\tilde{s}_i$ be the probability that a randomly chosen hyperedge joins at least two communities and $C^{(i)}$ is one of them. Note that for $s_i$ defined as in Lemma \ref{lemma:mod_hyper_gen} (i.e., the probability that a randomly chosen hyperedge has at least one vertex in $C^{(i)}$) we get $s_i = \tilde{s}_i + p_i$. By Lemma \ref{lemma:mod_hyper_gen} we get that with high probability
	\begin{equation} \label{eq:main}
		\begin{split}
			\lim_{t \rightarrow \infty} q^*(G) & \geqslant \sum_{i=1}^{r} p_i - \sum_{i=1}^{r} \sum_{\ell \geqslant 1} a_{\ell}\left(\frac{(d-1)\tilde{s}_i+dp_i}{\delta}\right)^{\ell}\\
			& = (1-\alpha) - \sum_{\ell \geqslant 1} \frac{a_{\ell}}{\delta^{\ell}} \sum_{i=1}^{r} ((d-1)\tilde{s}_i+dp_i)^{\ell} \\
			& = (1-\alpha) - \frac{a_1}{\delta}\left((d-1) \sum_{i=1}^{r} \tilde{s}_i + d \sum_{i=1}^{r} p_i \right) - \sum_{\ell \geqslant 2} \frac{a_{\ell}}{\delta^{\ell}} \sum_{i=1}^{r} ((d-1)\tilde{s}_i+dp_i)^{\ell}.
		\end{split}
	\end{equation}
	Now, by $r_k$ denote the probability that a randomly chosen hyperedge joins exactly $k$ communities. Note that
	\begin{equation} \label{eq:si}
		\sum_{i=1}^{r} \tilde{s}_i = 2r_2 + 3r_3 + \ldots + dr_d \leqslant d(1 - \sum_{i=1}^{r} p_i) = d \alpha.
	\end{equation}
	Thus
	\begin{equation} \label{eq:a1}
		\begin{split}
			\frac{a_1}{\delta}\left((d-1) \sum_{i=1}^{r} \tilde{s}_i + d \sum_{i=1}^{r} p_i \right) & \leqslant \frac{a_1}{\delta}\left((d-1) d \alpha + d (1-\alpha) \right) \\
			& = a_1 \left(\frac{d}{\delta}\right)((d-2)\alpha+1).
		\end{split}
	\end{equation}
	Moreover,
	\[
	\begin{split}
		\sum_{\ell \geqslant 2} \frac{a_{\ell}}{\delta^{\ell}} \sum_{i=1}^{r} ((d-1)\tilde{s}_i & +dp_i)^{\ell} = \sum_{\ell \geqslant 2} \frac{a_{\ell}}{\delta^{\ell}} \sum_{i=1}^{r} \sum_{k=0}^{\ell}{\ell \choose k}((d-1)\tilde{s}_i)^k (dp_i)^{l-k} \\
		& = \sum_{\ell \geqslant 2} \frac{a_{\ell}}{\delta^{\ell}} \sum_{k=0}^{\ell} {\ell \choose k}(d-1)^{k} d^{l-k} \sum_{i=1}^{r} \tilde{s}_i^k p_i^{l-k} \\
		& \leqslant \sum_{\ell \geqslant 2} \frac{a_{\ell}}{\delta^{\ell}} \sum_{k=0}^{\ell} {\ell \choose k}(d-1)^{k} (d\beta )^{l-k} \sum_{i=1}^{r} \tilde{s}_i^k \\
		& = \sum_{\ell \geqslant 2} \frac{a_{\ell}}{\delta^{\ell}} \left( r (d \beta)^{\ell} + \sum_{k=1}^{\ell} {\ell \choose k}(d-1)^{k} (d\beta )^{l-k} \sum_{i=1}^{r} \tilde{s}_i^k \right) \\
		& \leqslant \sum_{\ell \geqslant 2} \frac{a_{\ell}}{\delta^{\ell}} \left( r (d \beta)^{\ell} + \sum_{k=1}^{\ell} {\ell \choose k}(d-1)^{k} (d\beta )^{l-k} \left(\sum_{i=1}^{r} \tilde{s}_i\right)^k \right).
	\end{split}
	\]
	Next, by (\ref{eq:si}) we get
	\begin{equation} \label{eq:suml2}
		\begin{split}
			\sum_{\ell \geqslant 2} \frac{a_{\ell}}{\delta^{\ell}} \sum_{i=1}^{r} ((d-1)\tilde{s}_i +dp_i)^{\ell} & \leqslant \sum_{\ell \geqslant 2} \frac{a_{\ell}}{\delta^{\ell}} \left( r (d \beta)^{\ell} + \sum_{k=1}^{\ell} {\ell \choose k}(d-1)^{k} (d\beta )^{l-k} (d \alpha)^k \right) \\
			& = \sum_{\ell \geqslant 2} \frac{a_{\ell}}{\delta^{\ell}} \left( (r-1) (d \beta)^{\ell} + \sum_{k=0}^{\ell} {\ell \choose k}((d-1)d \alpha)^{k} (d\beta )^{l-k} \right) \\
			& = \sum_{\ell \geqslant 2} \frac{a_{\ell}}{\delta^{\ell}} \left( (r-1) (d \beta)^{\ell} + ((d-1)d \alpha + d \beta)^{\ell} \right) \\
			& = \sum_{\ell \geqslant 2} a_{\ell}\left(\frac{d}{\delta}\right)^{\ell}\left((r-1)\beta^{\ell} + ((d-1)\alpha + \beta)^{\ell}\right).
		\end{split}
	\end{equation}
	Finally, plugging (\ref{eq:a1}) and (\ref{eq:suml2}) to (\ref{eq:main}) we get that with high probability
	\[
	\begin{split}
		\lim_{t \rightarrow \infty} & q^*(G) \geqslant  \\
		& \geqslant 1-\alpha - a_1 \left(\frac{d}{\delta}\right)((d-2)\alpha+1) - \sum_{\ell \geqslant 2} a_{\ell}\left(\frac{d}{\delta}\right)^{\ell}\left((r-1)\beta^{\ell} + ((d-1)\alpha + \beta)^{\ell}\right).
	\end{split}
	\]\qed
\end{proof}

\end{document}